%% file: main-jmlr.tex
\documentclass[twoside,11pt]{article}

\usepackage{jmlr2e}
\input{macro-jmlr}
\usepackage{booktabs} 
\usepackage[symbol]{footmisc}
\renewcommand{\thefootnote}{\fnsymbol{footnote}}
\usepackage{algpseudocode,algorithm,algorithmicx}
\usepackage{multirow}
\usepackage{enumitem}

\ShortHeadings{Multi-scale Online Learning
and its Application to Online Auctions}{Bubeck, Devanur, Huang and Niazadeh}
\firstpageno{1}
\begin{document}
\title{Multi-scale Online Learning and its Applications to Online Auctions}

\author{\name S\'{e}bastien Bubeck \email sebubeck@microsoft.com  \\
       \addr Microsoft Research,\\
       1 Microsoft Way,\\
       Redmond, WA 98052, USA.
       \AND
       \name Nikhil Devanur \email nikdev@microsoft.com   \\
       \addr Microsoft Research,\\
       1 Microsoft Way,\\
       Redmond, WA 98052, USA.
	   \AND       
	   \name Zhiyi Huang \email zhiyi@cs.hku.hk   \\
       \addr Department of Computer Science,\\
       The University of Hong Kong,\\
       Pokfulam, Hong Kong.
       \AND
       \name Rad Niazadeh \email rad@cs.stanford.edu \\
       \addr Department of Computer Science,\\
       Stanford University,\\
       Stanford, CA 94305, USA.
       }

\editor{Csaba Szepesvari}
\maketitle

\begin{abstract}
We consider revenue maximization in online auction/pricing problems. 
A seller sells an identical item in each period to a new buyer, or a new set of buyers. 
For the online pricing problem, we show regret bounds that scale with the \emph{best fixed price}, rather than the range of the values. 
We also show regret bounds that are \emph{almost scale free}, and match the offline sample complexity, 
when comparing to a benchmark that requires a \emph{lower bound on the market share}. 
These results are obtained by generalizing the classical learning from experts and multi-armed bandit problems to their \emph{multi-scale} versions. 
In this version, the reward of each action is in a  \emph{different range}, 
and the regret with respect to a given action scales with its \emph{own range}, rather than the maximum range.
\end{abstract}
\begin{keywords}
online learning, multi-scale learning, auction theory, bandit information, sample complexity\footnote[2]{Following the theoretical computer science convention, we used \textbf{\emph{alphabetical author ordering.}}}
\end{keywords}

\input{introduction-ML}

\input{multi-scale}

\input{bandit}

\input{auction}

\input{symmetric}
\input{conclusion}

\begin{acks}
	
\end{acks}
\bibliography{bibliography}

\appendix
\input{appendix}

\end{document}

%% file: macro-jmlr.tex
\usepackage[T1]{fontenc}
\usepackage[utf8]{inputenc}
\usepackage[dvipsnames,usenames,table]{xcolor}
\usepackage{ifthen}
\usepackage{algpseudocode,algorithm,algorithmicx}

\algrenewcommand\algorithmicrequire{\textbf{Precondition:}}
\algrenewcommand\algorithmicensure{\textbf{Postcondition:}}
\usepackage{graphicx,color}
\usepackage{csquotes}

\usepackage{fancybox}
\usepackage{lipsum}
\usepackage{array,float}
\usepackage{url}
\usepackage{amstext,amssymb,amsmath}
\usepackage{mathrsfs}
\usepackage{dsfont}
\usepackage{pstricks, pst-tree, pst-node}
\usepackage{enumitem}
\usepackage{mathtools}
\usepackage{hyperref}

\usepackage{mathtools}
\usepackage{natbib}

\newcommand{\opt}{G_\textsc{max}}
\newcommand{\optdelta}{\opt(\delta)}
\newcommand{\alg}{G_\textsc{alg}}
\newcommand{\talg}{\widetilde{G}_\textsc{alg}}
\newcommand{\rev}{\textsc{rev}}
\newcommand{\regret}[1]{\textsc{regret}_{#1}}

\newcommand{\reward}[1]{G_{#1}}
\newcommand{\treward}[1]{\widetilde{G}_{#1}}

\newcommand{\ex}[1]{\mathbb{E}\left[#1\right]}

\newcommand{\eps}{\epsilon}

\newcommand{\argmin}[2]{\underset{#1}{\textrm{argmin}}\left(#2\right)}

\newcommand{\rew}[2]{g_{#1}(#2)}
\newcommand{\trew}[2]{\tilde{g}_{#1}(#2)}
\newcommand{\rewb}[1]{\mathbf{{g}}({#1})}
\newcommand{\trewb}[1]{\mathbf{\tilde{g}}({#1})}
\newcommand{\p}[2]{p_{#1}(#2)}
\newcommand{\tp}[2]{\tilde{p}_{#1}(#2)}
\newcommand{\pb}[1]{\mathbf{p}({#1})}
\newcommand{\tpb}[1]{\mathbf{\tilde{p}}({#1})}

\newcommand{\qb}{\mathbf{q}}
\newcommand{\mub}{\boldsymbol{\mu}}

\newcommand{\pib}{\boldsymbol{\pi}}

\newcommand{\w}[2]{w_{#1}(#2)}
\newcommand{\wb}[1]{\mathbf{w}({#1})}
\newcommand{\val}[2]{v_{#1}(#2)}
\newcommand{\vals}[1]{\mathbf{v} (#1)}
\newcommand{\pbest}{p^*}
\newcommand{\ibest}{i^*}
\newcommand{\imin}{i_{\min}}

\newcommand{\range}{c}
\newcommand{\actions}{\ensuremath{A}}

\newcommand{\expset}{\ensuremath{A}}

\newcommand{\naturals}{\mathbb{N}}
\newcommand{\one}{\vec{\mathbf{1}}}

\newcommand{\defeq}{:=}
\newcommand{\myfrac}[2]{{#1}/{#2}}

%% file: introduction-ML.tex
\section{Introduction}
\label{sec:intro}
Consider the following revenue maximization problem in a repeated setting, 
called the \emph{online posted pricing} problem. 
In each period, the seller has a single item to sell, and a new prospective buyer. 
The seller offers to sell the item to the buyer at a given price; 
the buyer buys the item if and only if the price is below his private valuation for the item. 
The private valuation of the buyer itself is never revealed to the seller. 
How should a monopolistic seller iteratively set the prices if he wishes to maximize his revenue? What if he also cares about the market share, i.e. the fraction of time periods at which the item is sold? 

Estimating price sensitivities and demand models in order to optimize 
revenue and market share is the bedrock of econometrics. 
The emergence of online marketplaces has enabled sellers to costlessly change prices, 
as well as collect huge amounts of data. 
This has renewed the interest in understanding best practices for data driven pricing. 
The extreme case of this when the price is updated for each buyer is the {online pricing} problem described above;
one can always use this for less frequent price updates. 
Moreover this problem is intimately related to the classical experimentation and estimation procedures. 

This problem has been studied from an\emph{ online learning} perspective, as a variant of 
the {\em multi-armed bandit} problem. In this variant, there is an arm for each possible price (presumably after an appropriate discretization). The revenue of each arm $p$ is either $p$ or zero, depending on whether the arriving value is at least equal to the price $p$ or smaller than the price $p$, respectively.
The total revenue of a pricing algorithm is then compared to 
the total revenue of the best fixed posted price in hindsight. 
The difference between the two, called the \emph{regret}, is then bounded from above. 
No assumption is made on the distribution of values; 
the regret bounds are required to hold for the \emph{worst case} sequence of values. 
\citet{blum2004online} assume that the buyer valuations are  in $[1,h]$, and show the following multiplicative plus additive bound on the regret: 
for any $\epsilon \in (0,1)$, the regret is at most 
$\epsilon $ times the revenue of the optimal price, plus
$O(\epsilon^{-2}h \log h \log\log h)$. 
\citet{blum2005near} show that the additive factor can be made to be $O(\epsilon^{-3}h  \log\log h)$, 
trading off a $\log h$ factor for an extra $\epsilon^{-1}$ factor.

An undesirable aspect of these bounds is that they scale \emph{linearly with} $h$;
this is particularly problematic when $h$ is an estimate and we might set it to be a generous upper bound on the range of prices we wish to consider. 
A typical use case is when the same algorithm is used for many different products,
with widely varying price ranges. We may not be able to manually tune the range for each product separately. 

One might wonder if this dependence on $h$ is unavoidable, as it seems to be reflected by the existing lower bounds for this problem in the literature (lower bounds are discussed later in the introduction with more details). Interestingly, in all of these lower-bound instances the best fixed price is equal to $h$ itself; Therefore, it is not clear whether this dependency on $h$ is required for instances where $h$ is only a pessimistic upper-bound on the best fixed price. We now ask the following question:

\begin{displayquote}
\textbf{Question}: \emph{do online learning algorithms exist for the online posted pricing problem, such that their regrets are proportional to the best fixed price instead of the highest value?}
\end{displayquote}

Standard off-the-shelf bounds allow regret to depend on the loss of the best arm instead of the worst case loss. However, even such bounds still depend linearly on the maximum range of all the losses, and thus they would not allow to replace $h$ by the best fixed price.

 Fortunately, in the online pricing problem the reward function of the arms is well structured. In particular, as a neat observation, the reward of the arm $p$ is upper-bounded by $p$ (and not only the maximum value). Can we use this structure in our favor to improve the standard regret bounds? We answer this question in the affirmative by the means of reducing the problem to a pure learning problem termed as  \emph{mutli-scale online learning}.

\subsection{Multi-scale online learning}
The main technical ingredients in our results are variants of the classical problems of learning from expert advice and multi-armed bandit. 
We introduce the multi-scale versions of these problems, 
where each action has its reward bounded in a different range. Here, we seek to design online learning algorithms that guarantee \emph{multi-scale regret} bounds, i.e. their regrets with respect to each certain action scales with the range of that particular action, instead of the maximum possible range. These guarantees are in contrast with the regret bounds of the standard versions, which scale with the maximum range. 
\begin{displayquote}
\textbf{Main result (informal):} \emph{we give algorithms for the full information and bandit information versions of the multi-scale online learning problem with multi-scale regret guarantees. }
\end{displayquote}
While we use these bounds mostly for designing online auctions and pricing mechanisms, we expect such bounds to be of independent interest. 

The main idea behind our algorithms is to use a tailored  variant of \emph{online (stochastic) mirror descent (OSMD)} \citep{bubeck2011introduction}. In this tailored version, the algorithm uses a \emph{weighted negative entropy} as the Legendre function (also known as the \emph{mirror map}), where the weight of each term $i$ (corresponding to arm $i$)  is actually equal to the range of that arm. More formally, assuming the range of arm $i$ is equal to $c_i$, our mirror descent algorithms  (Algorithm~\ref{alg:MSMW} for full information, and Algorithm~\ref{alg:Bandit-MSMW} for the bandit information) use the following mirror map: 
$$
F(x)=\sum_{\textrm{arms $i$}}c_i\cdot x_i\ln(x_i)
$$ 

Intuitively speaking, these algorithms take into account different ranges for different arms by first normalizing the reward of each arm by its range (i.e.  divide the reward of arm $i$ by its corresponding range $\range_i$), and then projecting the updated weights by performing a \emph{smooth multi-scale projection} into the simplex. This projection is an instant of the more general Bregman projection~\citep{bubeck2011introduction} for the special case of weighted negative entropy as the mirror map. The mirror descent framework then gives regret bounds in terms of a ``local norm'' as well as an ``initial divergence'', which we then bound differently for each version of the problem. In the technical sections we highlight how the subtle variations arise as a result 
of different techniques used to bound these two terms. 
 
While our algorithms have the style of the multiplicative weights update  (up to a normalization of the rewards), the smooth projection step at each iteration makes them drastically different.  To shed some insight on this projection step, which plays an important role in our analysis, consider a very special case of the problem where the reward of each arm $i$ is deterministically equal to $c_i$. The multiplicative weights  algorithm picks arm $i$ with a probability proportional to $\exp{(c_i)}$. However, as it is clear from the description of Algorithm~\ref{alg:MSMW}, our algorithm uniformly scales the weight of each arm first. Then, in the projection step the weight of each arm $i$ is multiplied by $\exp{(-\tfrac{\lambda^*}{c_i})}$ for some parameter $\lambda^*$. Hence, arm $i$ will be sampled with a probability proportional to $\exp{(-\tfrac{\lambda^*}{c_i})}$ (which is a smooth approximation to $i^*=\textrm{argmax}~c_i$, but in a different way compared to the vanilla multiplicative weights).

The multi-scale versions exhibit subtle variations that do not appear in the standard versions. 
First of all, our applications to auctions and pricing have non-negative rewards, and 
this actually makes a difference. 
For both the expert and the bandit versions, the minimax regret bounds for non-negative rewards are \emph{provably better} than those when rewards could be negative.
Further, for the bandit version, we can prove a better bound if we only 
require the bound to hold with respect to the \emph{best} action, rather than \emph{all} actions (for non-negative rewards). The various regret bounds and comparison to standard bounds are summarized in Tables~\ref{tab:multi-scale}.
\begin{table}[ht]
\renewcommand\arraystretch{1.5}
\renewcommand{\thefootnote}{\fnsymbol{footnote}}
 	\small
\centering
 	\begin{tabular}{|c|c|c|c|}
 		\hline
 		& \multirow{2}{*}{\shortstack{Standard \\regret bound $O(\cdot)$}} & \multicolumn{2}{|c|}{Multi-scale bound (this paper)} \\
 		\cline{3-4}
 		& & \shortstack{Upper bound  $O(\cdot)$ }& \shortstack{Lower bound  $\Omega(\cdot)$}\\
 		\hline
 		Experts/non-negative  & 
 		$ \range_{\max}\sqrt{T\log(k)}$~~\footnotemark[1] & 
 		$ \range_{i}\sqrt{T\log(kT)} $ & $ \range_{i}\sqrt{T\log(k)} $ \\
 		\hline
 		 \multirow{2}{*}{Bandits/non-negative} &  \multirow{2}{*}{$\range_{\max}\sqrt{Tk}$~~\footnotemark[2]} & $ \range_{i} T^{\frac{2}{3}}(k\log(kT))^{\frac{1}{3}} $ &  $c_i\sqrt{TK}$\\
 		 \cline{3-4} &&$\range_{i^*}\sqrt{Tk\log(k)} $, $i^*$ is the best action & -\\
 		\hline
 		Experts/symmetric & 	$\range_{\max}\sqrt{T\log(k)}$~~\footnotemark[1] & $ \range_{i}\sqrt{T\log (k\cdot \frac{\range_{\max}}{\range_{\min}})} $ & $\range_{i}\sqrt{T\log(k)} $ \\
 		\hline
 		Bandits/symmetric & 	$\range_{\max}\sqrt{Tk}$~~\footnotemark[2] & $ \range_{i}\sqrt{Tk\cdot \frac{\range_{\max}}{\range_{\min}}\log (k T\cdot \frac{\range_{\max}}{\range_{\min}})} $ & 
 		 $ \range_{i}\sqrt{Tk\cdot \frac{\range_{\max}}{\range_{\min}}}$ \\
 		\hline
 	\end{tabular}
 		
 	\smallskip
 	
\noindent\par
 	\begin{minipage}{0.6\textwidth}
 		{
 			\center
 			\footnotesize
 			\footnotemark[1]~~{\citet{freund1995desicion}};\quad
 			\footnotemark[2]~~{\citet{audibert2009minimax}}.
 		}
 	\end{minipage}
 	
 	\smallskip
 	
 	\caption{Pure-additive regret bounds for non-negative rewards, i.e. when reward of any action $i$ at any time is in $[0,c_i]$, and symmetric range rewards, i.e. when reward of any action $i$ at any time is in $[-c_i,c_i]$ (suppose $T$ is the time horizon, $\actions$ is the action set, and $k$ is the number of actions).}
 	\label{tab:multi-scale}
 \end{table}

\subsection{The implications for online auctions and pricing}

As a direct application of our multi-scale online learning framework, somewhat surprisingly, 

\begin{displayquote}
\textbf{Second contribution:} \emph{we show that we can get regret proportional to the best fixed price instead of the highest value for the online posted pricing problem.}
\end{displayquote}
(i.e., we can replace $h$ by the best fixed price, which is used in the definition of the benchmark). In particular, we show that the additive bound can be made to be $O(\epsilon^{-2} p^* \log h)$, where $p^*$ is the best fixed price in hindsight. This allows us to use a very generous estimate for $h$ and let the algorithm adapt to the actual range of prices; we only lose a $\log h$ factor.
The algorithm balances exploration probabilities of different prices carefully and automatically zooms in on the relevant price range. 
This does not violate known lower bounds, since in those instances $p^*$ is close to $h$. 

\citet{bar2002incentive}, \citet{blum2004online}, and \citet{blum2005near} also consider the ``full information''  version of the problem, or what we call the \emph{online (single buyer) auction} problem, where 
the valuations of the buyers are revealed to the algorithm after the buyer has made a decision. 
Such information may be available in a context where the buyers have to bid for the items, and are awarded the item if their bid is above a hidden price. 
In this case, the additive term can be improved to $O(\epsilon^{-1} h \log (\epsilon^{-1}))$, which is tight. 
Once again, by a reduction to multi-scale online learning, \emph{we show that  $h$  can be replaced with} $p^*$; 
in particular, we show that the additive term can be made to be $O(\epsilon^{-1}p^* \log(h\epsilon^{-1}))$.

\subsection{Purely multiplicative bounds and sample complexity}
 The regret bounds mentioned above can be turned into a purely multiplicative factor in the following way: for any $\epsilon > 0$, the algorithm is guaranteed to get a $1-O(\epsilon)$ fraction of the best fixed price revenue, 
 provided the number of periods $T \geq E/\epsilon,$
 where $E$ is the additive term in the regret bounds above. 
 This follows from the observation that a revenue of $T$ is a lower bound on the best fixed price revenue. 
Define the number of periods required to get a $1-\epsilon$ multiplicative approximation (as  a function of $\epsilon$) to be  the \emph{convergence rate} of the algorithm. 
 
A $1-\epsilon$ multiplicative factor is also the target in the recent line of work, on the \emph{sample complexity} of auctions, started by  \citet{balcan2008red, elkind2007des,Dhangwatnotai2014revenue,ColeR14}. 
(We give a more comprehensive discussion of this line of work in Section~\ref{sec:related-work}.)
Here,\emph{ i.i.d.} samples of the valuations are given from a \emph{fixed but unknown distribution}, and the goal is to find a price such that its revenue with respect to\ the hidden distribution is a 
$1-\epsilon$ fraction of the optimum revenue for this distribution. 
The sample complexity is the minimum number of samples needed to guarantee this (as a function of $\epsilon$).

The sample complexity and the convergence rate (for the full information setting) are closely related to each other. 
The sample complexity is always smaller than the convergence rate: the problem is easier because of the following. 
\begin{enumerate}
	\item The valuations are i.i.d. in the case of sample complexity, whereas they can be arbitrary (worst case) in the case of convergence rate.  
	\item Sample complexity corresponds to an offline problem: you get all the samples at once. Convergence rate corresponds to an online problem: you need to decide what to do on a given valuation without knowing what valuations arrive in the future.
\end{enumerate}
This is formalized in terms of an \emph{online to offline reduction} [folklore]
which shows  that a convergence rate upper bound can be automatically translated to a sample complexity upper bound.
This lets us convert sample complexity lower bounds into lower bounds on the convergence rate, and in turn into lower bounds on the additive error $E$  in an additive plus multiplicative regret bound. 
For example, the additive error for the  online auction problem (and hence also for the posted pricing problem\footnote{
	We conjecture that the lower bound for the posted pricing problem should be worse by a factor of $\epsilon^{-1}$, since one needs to explore about $\epsilon^{-1}$ different prices. 
}) cannot be $o(h\epsilon^{-1})$~\citep{huang2015making}. 
Moreover, it is insightful to compare convergence rates we show with {\em the best known sample complexity upper bound;  proving better convergence rates would mean improving these bounds as well}. 

A natural target convergence rate for a problem is therefore the corresponding sample complexity, but achieving this is not always trivial. 
In particular, we consider an interesting version of the sample complexity bound for auctions, for which no analogous convergence rate bound is known in the literature. This version takes into account \emph{both revenue and market share}, and gets sample complexity bounds that are \emph{scale free}; 
there is no dependence on $h$, which means it works for unbounded valuations! For any $\delta \in (0,1)$, the best fixed price benchmark is relaxed to ignore those prices whose market share (which is equivalent to the probability of sale) is below a 
$\delta $ fraction; as $\delta$ increases the benchmark is lower.
This is a meaningful benchmark since in many cases revenue is not the only goal, even if you are a monopolist.
A more reasonable goal is to maximize revenue subject to 
the constraint that the market share is above a certain threshold.
What is more, this gives a sample complexity of $O(\epsilon^{-2}\delta^{-1}\log (\delta^{-1} \epsilon^{-1}))$~\citep{huang2015making}. 
In fact $\delta$ can be set to $h^{-1}$ without loss of generality, when the values are in $[1,h]$,\footnote{\label{fn1}
	When the values are in $[1,h]$, we can guarantee a revenue of $T$ by posting  a price of 1, and to beat this, any other price (and in particular a price of $h$) would have to sell at least $T/h$ times.
} 
and the above bound then matches the sample complexity with respect to the best fixed price revenue.
In addition, this bound gives a precise interpolation:  
as the target market share $\delta$ increase, the number of samples needed decreases almost linearly. 

\begin{displayquote}
\textbf{Third contribution:} \emph{we show a convergence rate that almost matches the above sample complexity, for the full information setting.}
\end{displayquote}
We have a mild dependence on $h$; the rate is proportional to $\log\log h$. 
Further, we also show a near optimal convergence rate for the online posted pricing problem.\footnote{
Unfortunately, we cannot yet guarantee that our online algorithm itself gets a market share of $\delta$, although we strongly believe that it does. Showing such bounds on the market share of the algorithm is an important avenue for future research.}

\paragraph{Multiple buyers:}
All of our results in the full information (online auction) setting extend to the multiple buyer model.
In this  model, in each time period, a new set of $n$ buyers competes for a single item. 
The seller runs a truthful auction that determines the winning buyer and his payment. 
The benchmark here is the set of all ``Myerson-type'' mechanisms. 
These are mechanisms that are optimal when each period has $n$ buyers of potentially different types, and the value of each buyer is drawn independently from a type dependent distribution.
In fact, our convergence rates also imply new sample complexity bounds for these problems 
(except that they are not computationally efficient). 
 
 The various bounds and comparisons to previous work are summarized in Tables~\ref{tab:opt} \& \ref{tab:opt-delta}. 
 
 \begin{table}[ht]
 
 	\renewcommand\arraystretch{1.5}
	\renewcommand{\thefootnote}{\fnsymbol{footnote}}
	\centering
 	\small
 	\resizebox{\textwidth}{!} {%

 	\begin{tabular}{|c|c|c|c|c|}
 		\hline
 		& \multirow{2}{*}{\shortstack{Lower bound }} & \multicolumn{3}{|c|}{Upper bound} \\
 		\cline{3-5}
 		& & \shortstack{Best known \\(Sample complexity) }& \shortstack{Best known \\ (Convergence rate)} & \shortstack{This paper \\(Thm.~\ref{thm:multadditivebounds}) }\\
 		\hline
 		Online single buyer auction & 
 		$\Omega\big(\frac{h}{\epsilon^2}\big)$~~\footnotemark[1] & 
 		$\tilde{O}\big(\frac{h}{\epsilon^2}\big)$ ~~\footnotemark[2] &
 		$\tilde{O}\big(\frac{h}{\epsilon^2}\big)$~~\footnotemark[2] & $\tilde{O}\big(\frac{\pbest}{\epsilon^2}\big)$ \\
 		\hline
 		Online posted pricing & $\Omega\big( \max\{ \frac{h}{\epsilon^2}, \frac{1}{\epsilon^3} \} \big)$~~\footnotemark[1]\footnotemark[4] & - & 
 		$\tilde{O}\big(\frac{h}{\epsilon^3}\big)$~~\footnotemark[2] & $\tilde{O}\big(\frac{\pbest}{\epsilon^3}\big)$ \\
 		\hline
 		Online multi buyer auction & $\Omega(\frac h {\epsilon^2})$~~\footnotemark[1] & 
 		$O(\frac{nh}{\epsilon^3})$ ~~\footnotemark[3] & - & $\tilde{O}\big(\frac{n h}{\epsilon^3}\big)$\\
 		\hline
 	\end{tabular}
 	}
 	\smallskip
 	
 	\noindent
 	\begin{minipage}{\textwidth}
 		{
 			\footnotesize
 			\footnotemark[1]~~{\citet{huang2015making}};~~
 			\footnotemark[2]~~{\citet{blum2004online}};~~
\footnotemark[3]~~{\citet{devanur2016sample,gonczarowski2016efficient,elkind2007des}};~~
\footnotemark[4]~~{\citet{kleinberg2003value}}.
 		}
 	\end{minipage}
 	
 	\smallskip
 	
 	\caption{Number of rounds/samples needed to get a $1-\epsilon$ approximation to the best offline price/mechanism. Sample complexity is for the offline case with i.i.d. samples from an unknown distribution. 
 		Convergence rate is for the online case with a worst case sequence. 
 		Sample complexity is always no larger than the convergence rate. 
 		Lower bounds hold for sample complexity too, except for the online posted pricing problem for which there is no sample complexity version. The additive plus multiplicative regret bounds are converted to convergence rates by dividing the additive error by $\epsilon$. 
 		In the last row, $n$ is the number of buyers. 
 	In the last column, 	$\pbest$ denotes the optimal price. }
 	
 	\vspace{2mm}
 	\label{tab:opt}
	
 	\begin{tabular}{|c|c|c|c|}
 		\hline
 		& \multirow{2}{*}{\shortstack{Lower bound \\(Sample complexity)}} & \multicolumn{2}{|c|}{Upper bound} \\
 		\cline{3-4}
 		& & \shortstack{Best known \\(Sample complexity) }& \shortstack{This paper \\(Thm.~\ref{thm:convergencerates}) }\\
 		\hline
 		Online single buyer auction & 
 		$\Omega\big(\frac{1}{\epsilon^2 \delta}\big)$~~\footnotemark[1] & 
 		$\tilde{O}\big(\frac{1}{\epsilon^2 \delta}\big)$~~\footnotemark[1] & $\tilde{O}\big(\frac{1}{\epsilon^2 \delta}\big)$ \\
 		\hline
 		Online posted pricing & $\Omega\big(\max\{ \frac{1}{\epsilon^2 \delta}, \frac{1}{\epsilon^3} \}\big)$~~\footnotemark[1]\footnotemark[2] & - & $\tilde{O}\big(\frac{1}{\epsilon^4 \delta}\big)$ \\
 		\hline
 		Online multi buyer auction & 	$\Omega\big(\frac{1}{\epsilon^2 \delta}\big)$~~\footnotemark[1] & - & $\tilde{O}\big(\frac{n}{\epsilon^3 \delta}\big)$\\
 		\hline
 	\end{tabular}
 	
 	\smallskip
 	
 	\noindent\par
 	\begin{minipage}{0.65\textwidth}
 		{
 			\footnotesize
 			\footnotemark[1]~~{\citet{huang2015making}};\quad
			\footnotemark[2]~~{\citet{kleinberg2003value}}.
 		}
 	\end{minipage}
 	
 	\smallskip
 	
 	\caption{Sample complexity \& convergence rate w.r.t.\  the opt mechanism/price with market share $\geq\delta$.}
 	\label{tab:opt-delta}
 \end{table}

\vspace{-2mm}

 \subsection{Other related work}
 \label{sec:related-work} 
The online pricing problem, also called \emph{dynamic pricing}, is a much studied topic, across disciplines such as operations research and management science \citep{talluri2006theory}, economics \citep{segal2003optimal}, marketing, and of course computer science. 
The multi-armed bandit approach to pricing is particularly popular. 
  See \citet{den2015dynamic} for a recent survey on various approaches to the problem. 
  
  \citet{kleinberg2003value} consider the online pricing problem, under the assumption that the values are in $[0,1]$,  and considered purely additive factors. 
  They showed that the minimax additive  regret is $\tilde{\Theta}(T^{2/3})$, where $T$ is the number of periods. 
  This is similar in spirit to regret bounds that scale with $h$, since one has to normalize the values 
  so that they are in $[0,1]$.  
   The finer distinction about the magnitude of the best fixed price is absent in this work. 
Recently, \citet{syrgkanis2017sample} also consider the online auction problem, with an emphasis on 
a notion of ``oracle based'' computational efficiency. 
They assume the values are all in $[0,1]$ and do not consider the scaling issue that we do;
 this makes their contribution orthogonal to ours. 

Starting with \citet{Dhangwatnotai2014revenue}, there has been a spate of recent results analyzing the sample complexity of pricing and auction problems. 
\citet{ColeR14} and \citet{devanur2016sample} consider multiple buyer auctions with regular distributions (with unbounded valuations) and give sample complexity bounds that are polynomial in $n$ and $\epsilon^{-1}$, where $n$ is the number of buyers. 
\citet{MR15} consider arbitrary distributions with values bounded by $h$, and gave bounds that are polynomial in $n,h,$ and $\epsilon^{-1}$. 
\citet{roughgarden2016ironing,huang2015making} give further improvements on the single- and multi-buyer versions respectively;
Tables~\ref{tab:opt} and \ref{tab:opt-delta} give a comparison of these results with our bounds, for the problems we consider. 
The dynamic pricing problem has also been studied when there are a given number of copies of the item to sell (limited supply) \citep{agrawal2014bandits,babaioff2015dynamic,badanidiyuru2013bandits,besbes2009dynamic}. 
There are also variants where the seller interacts with the same buyer repeatedly, and the buyer can strategize to influence his utility in the future periods \citep{amin2013learning}. 

\citet{foster}  also consider the multi-scale online learning problem motivated by a model selection problem.
 They consider additive bounds, for the symmetric case, for full information, but not bandit feedback. 
Their regret bounds are not comparable to ours in general; 
our bounds are better for the pricing/auction applications we consider, 
and their bounds are better for their application.


\paragraph{Organization} We start in Section~\ref{sec:learning} by showing regret upper bounds for the multi-scale experts problem with non-negative rewards (Theorem \ref{thm:expert-regret-bound-rewards}). The corresponding upper bounds for the bandit version are in section \ref{sec:bandit} (Theorem~\ref{thm:bandit-regret-bound}). In Section~\ref{sec:auction} we show how the multi-scale regret bounds (Theorems \ref{thm:expert-regret-bound-rewards} and \ref{thm:bandit-regret-bound}) 
  imply the  corresponding bounds for the auction/pricing problems (Theorems \ref{thm:multadditivebounds} and \ref{thm:convergencerates}). Finally, the regret (upper and lower) bounds for the symmetric range are discussed in Section~\ref{sec:symmetric} (Theorems~~\ref{thm:expert-regret-bound-symmetric}, \ref{thm:expert-log-dependency-symmetric}, \ref{thm:bandit-regret-bound-symmetric}, and \ref{thm:bandit-lower-bound-symmetric}).

%% file: multi-scale.tex
\section{Full Information Multi-scale Online Learning}

\label{sec:learning}
We consider a variety of online algorithmic problems that are all parts of the \emph{multiscale online learning} framework. We start by defining this framework, in which different actions have different ranges. We exploit this structure and express our results in terms of \emph{action-specific} regret bounds for this general problem. To obtain these results, we use a variant of online mirror descent and propose a multiplicative-weight update style learning algorithm for our problem, termed as \emph{Multi-Scale Multiplicative-Weight (MSMW)} algorithm. 

Next, we investigate the single buyer auction problem (or equivalently the full-information single buyer dynamic pricing problem) as a canonical application, and show how to get multiplicative cum additive approximations here by the help of the multi-scale online learning framework. To show the tightness of our bounds, we compare the convergence rate of our dynamic pricing with the sample complexity of a closely related offline problem, i.e. the near optimal Bayesian revenue maximization from samples~\citep{ColeR14}.

\subsection{The  framework}

\label{sec:multi-scale-nonnegative-results}

Our full-information multi-scale online learning framework is basically the classical learning from expert advice problem. The main difference is that the \emph{range} of rewards of different experts could be different.  More formally, suppose there is a set of actions $\actions$.\footnote{We use the terms experts, arms and actions interchangeably in this paper.} The online problem proceeds in $T$ rounds, where in each round $t \in [T]: $\footnote{
	We use the notation  $[n] :=  \{ 1,2,\ldots,n\} $, for any $n \in \naturals.$}  
\begin{itemize}[itemindent=-2mm]
	\item The adversary picks a reward function $\rewb{t}$, where $\rew{i}{t}$ is the reward of action $i$.
	\item The algorithm picks an action $i_t \in \actions$ simultaneously.
	\item Then the algorithm gets the reward $\rew{i_t}{t}$ and observes the entire reward function $\rewb{t}$.
	
\end{itemize}
The total reward of the algorithm is denoted by 
$$\textstyle \alg \defeq \sum_{t=1}^{T} \rew{i_t}{t} . $$
The standard ``best fixed action''  benchmark is  
$$\textstyle
\opt \defeq\max_{i \in \actions} \sum_{t=1}^{T} \rew{i}{t}.$$	
We further assume that the action set is finite. Without loss of generality, if the action set is of size $k$, we identify $\actions = [k]$. The reward $\rewb{t}$ is such that for all $i \in \actions$, $\rew{i}{t} \in [0,\range_i]$, where $\range_i\in\mathbb{R}_+$  is the \emph{range} of action $i$. 
\subsection{Multi-scale regret bounds}
\label{sec:full-info-regret-simple}
We prove action-specific regret bounds, which we call also \emph{multi-scale regret guarantees}.  Towards this end, we define the following quantities. 
\begin{eqnarray}
\reward{i} & \defeq & \textstyle \sum_{t\in[T]}\rew{i}{t}~,\\[0.8ex]
\regret{i} & \defeq & \reward{i}-\alg~.
\end{eqnarray}
The regret bound w.r.t. action $i$, i.e., an upper bound on $\ex{\regret{i}}$, depends on the range $\range_i$, as well as 
any \emph{prior} distribution $\pib$ over the action set $\actions$;
this way, we can handle countably many actions. 
Let $\range_{\min} = \inf_{i \in \actions} \range_i $ and $\range_{\max} = \sup_{i \in \actions} \range_i $ (if applicable) be the minimum and the maximum range. 
We first state a version of the regret bound which is parameterized by $\epsilon > 0$; 
such bounds are stronger than $\sqrt{T}$ type bounds which are more standard. 
\begin{theorem}[Main Result]
	\label{thm:expert-regret-bound-rewards}
There exists an algorithm for the full-information multi-scale online learning problem that	 takes as input any distribution $\pib$ over $\expset$, the ranges $c_i, ~\forall ~i \in A$ and a parameter $0 < \epsilon \le 1$,  and satisfies:
	\begin{equation}
	\forall i \in \actions:~~\ex{\regret{i}} \leq \epsilon \cdot \reward{i} + O \left( \frac{1}{\epsilon} \log \big( \frac{1}{\epsilon \pi_i} \big) \cdot \range_i  \right)
	\end{equation}
\end{theorem}
Compare this to what you get by using the standard analysis for the experts problem~\citep{arora2012multiplicative}, where the second term in the regret bound  is 
$O \big( \frac{1}{\epsilon} \log (k ) \cdot \range_{\max} \big)$.
Choosing $\pib$ to be the uniform distribution in the above theorem gives $O \big( \frac{1}{\epsilon} \log \big( \frac{k}{\epsilon} \big) \cdot \range_i  \big)$. Also, one can compare the pure-additive version of this bound with the classic pure-additive regret bound $O\big( \range_{\max}\cdot \sqrt{T\log(k)}\big)$ for the experts problem by setting $\epsilon=\sqrt{\frac{\log (kT)}{T}}$ (Corollary~\ref{cor:expert-regret-bound-rewards-additive}).
\begin{corollary}
	\label{cor:expert-regret-bound-rewards-additive}
	There exists an algorithm for the full-information multi-scale online learning problem that takes as input the ranges $c_i, ~\forall ~i \in A$,  and satisfies:
	\begin{equation}
	\forall i \in \actions:~~\ex{\regret{i}} \leq {O} \left( \range_{i}\cdot\sqrt{T\log(kT)} \right)
	\end{equation}
\end{corollary}
\begin{remark}
We should assert that in a multi-scale regret guarantee, we provide a separate regret bound for each action, where the bound on the regret of action $i$ only scales linearly with $\range_{i}$. This type of guarantee should ``not'' be mistaken as a bound on the worst action.
\end{remark}

Here is the map of the rest of this section. In Section~\ref{sec:MSMW} we propose an algorithm that exploits the reward structure, and later in Section~\ref{sec: OMD-proof} we show how this algorithm is an online mirror descent with weighted negative entropy as its mirror map. For reward-only instances, we prove the regret bound in Section~\ref{sec:expert-reward-only}. We finally turn our attention to the single buyer online auction problem in Section~\ref{sec:canonical-pricing}.

\subsection{Multi-Scale Multiplicative-Weight (MSMW) algorithm}
\label{sec:MSMW}
We achieve our regret bound in Theorem~\ref{thm:expert-regret-bound-rewards} by using the MSMW algorithm (Algorithm~\ref{alg:MSMW}). The main idea behind this algorithm is to take into account different ranges for different experts, and therefore:
\begin{enumerate}
\item We normalize the reward of each expert accordingly, i.e.  divide the reward of expert $i$ by its corresponding range $\range_i$;
\item We project the updated weights by performing a \emph{smooth multi-scale projection} into the simplex: the algorithm finds a $\lambda^*$ such that multiplying the current weight of each expert $i$ by $\exp{(-\tfrac{\lambda^*}{c_i})}$ makes a probability distribution over the experts. It then uses this resulting probability distribution for sampling the next expert. 
\end{enumerate}

\begin{algorithm}[h]
\caption{MSMW}
\label{alg:MSMW}
\begin{algorithmic}[1]
	\State{\textbf{input}}~ initial distribution $\mub$ over $\expset$, learning rate $0 < \eta \le 1$.
  \State{\textbf{initialize}}~ $\pb{1}$ such that $\p{i}{1} = \mu_i$ for all $i \in \actions$.  
	\For{$t=1,\dots,T$}
		\State Randomly pick an action drawn from $\pb{t}$, and observe $\rewb{t}$. 
		\State $\forall i \in \actions:~~\w{i}{t+1}\leftarrow\p{i}{t}\cdot \exp(\eta\cdot\frac{\rew{i}{t}}{\range_i})$. 
		\State Find $\lambda^*$ (e.g., binary search) s.t. $\sum_{i \in \actions}\w{i}{t+1}\cdot\exp(-\frac{\lambda^*}{\range_i})=1$.
		\State $\forall i \in \actions:~~\p{i}{t+1}\leftarrow \w{i}{t+1}\cdot\exp(-\frac{\lambda^*}{\range_i}).$
  \EndFor		
\end{algorithmic}
\end{algorithm}

\subsection{Equivalence to online mirror descent  with weighted negative entropy}
\label{sec: OMD-proof}
\input{OMD-NegEntropy}

\subsection{Regret analysis for non-negative rewards}
\label{sec:expert-reward-only}
\begingroup
\def\thetheorem{\ref{thm:expert-regret-bound-rewards}}
\begin{theorem}
	
There exists an algorithm for the full-information multi-scale online learning problem that	 takes as input any distribution $\pib$ over $\expset$, the ranges $c_i, ~\forall ~i \in A$ and a parameter $0 < \epsilon \le 1$,  and satisfies:
	\begin{equation}
	\forall i \in \actions:~~\ex{\regret{i}} \leq \epsilon \cdot \reward{i} + O \left( \frac{1}{\epsilon} \log \big( \frac{1}{\epsilon \pi_i} \big) \cdot \range_i  \right)
	\end{equation}
\end{theorem}
\addtocounter{theorem}{-1}
\endgroup
\begin{proof}[\textbf{of Theorem~\ref{thm:expert-regret-bound-rewards}}]
Suppose $\imin$ is an action with the minimum $\range_{i}$.
Let $\mub = (1 - \eta) \cdot \mathbf{1}_{\imin} + \eta \cdot \pib$, and let $\qb = (1 - \eta) \cdot \mathbf{1}_{i} + \eta \cdot \pib$ in Proposition~\ref{prop:expert-regret-bound-general}.
If $i \ne \imin$, we get that (note that $\mu_j = q_j$ for any $j \ne i, \imin$):
\begin{align*}
(1 - \eta) \cdot \reward{i} + \eta \cdot \sum_{j \in \actions} \pi_j \cdot \reward{j} - \ex{\alg }
& \leq 
\eta\cdot\ex{\alg} + \frac{1}{\eta} \cdot \range_i \cdot \bigg( q_i \ln \big( \frac{q_i}{\mu_i} \big) - q_i + \mu_i \bigg) \\
& \quad\quad
+ \frac{1}{\eta} \cdot \range_{\imin} \cdot \bigg( q_{\imin} \ln \big( \frac{q_{\imin}}{\mu_{\imin}} \big) - q_{\imin} + \mu_{\imin} \bigg)
\end{align*}

By $1 \ge q_i > \mu_i \ge \eta \pi_i$, the second term on the RHS is upper bounded as:
\[
\frac{1}{\eta} \cdot \range_i \cdot \bigg( q_i \ln \big( \frac{q_i}{\mu_i} \big) - q_i + \mu_i \bigg)
\le 
\frac{1}{\eta} \cdot \range_i \cdot \ln \big( \frac{1}{\eta \pi_i} \big)
\]

Similarly, by $ 1 \ge \mu_{\imin} > q_{\imin} \ge 0$, the third term on the RHS is upper bounded as
\[
\frac{1}{\eta} \cdot \range_{\imin} \cdot \bigg( q_{\imin} \ln \big( \frac{q_{\imin}}{\mu_{\imin}} \big) - q_{\imin} + \mu_{\imin} \bigg) 
\le
\frac{1}{\eta} \cdot \range_{\imin} 
\le
\frac{1}{\eta} \cdot \range_{i} 
\]

Finally, note that $\reward{j} \ge 0$ for all $j \in \actions$ in reward-only instances.
So the LHS is lower bounded by 
\[
(1 - \eta) \cdot \reward{i} - \ex{\alg}
=
(1 - \eta) \cdot \regret{i} - \eta \cdot \ex{\alg}.
\]

Putting all this  together, we get that 
\[
\ex{\regret{i}}
\leq
\frac{2\eta}{1 - \eta} \cdot \ex{\alg} + O \bigg( \frac{1}{\eta}\ln \big( \frac{1}{\eta \pi_i} \big) \cdot \range_i \bigg) 
\le
3\eta \cdot \ex{\alg} + O \bigg( \frac{1}{\eta}\ln \big( \frac{1}{\eta \pi_i} \big) \cdot \range_i \bigg) 
~.
\]

The theorem then follows by choosing $\eta = \frac{\epsilon}{3}$ and rearranging terms.
\end{proof}

\subsection{A canonical application: online single buyer auction}
\label{sec:canonical-pricing}
\paragraph{The setup.} The simple auction design problem that we consider is as follows. There is a seller with infinite identical copies of an item. Buyers arrive over time. At each round, the seller picks a \emph{price} and the arriving buyer reports her \emph{value}. If the value is no less than the price, the trade happens; money goes to the seller and the copy of the item goes to the arriving buyer. The goal is to maximize the revenue of the seller.

Formally, we look at this problem as an instance of the full information multi-scale online learning framework; The action set is  $\actions = [1,h]$.~\footnote{Here, we allow an infinite action set. Later, we show how to discretize to get around this issue.}  The reward function is such that at round $t$ the adversary (i.e. the arriving buyer) picks a value $v(t) \in [1,h]$ and for any price $p\in \actions$ picked by the seller (i.e. the algorithm), the reward is $\rew{p}{t} := p \cdot \mathbf{1}(v(t) \geq p)$. This is a full information setting, because the value $v(t)$ is revealed to the algorithm after each round $t$. 

\paragraph{The additive/multiplicative approximation.} In order to obtain a $(1 -\epsilon)$-approximation of the optimal revenue, i.e. the revenue of the best fixed price $p^*$ in hindsight,  it suffices to consider prices of the form $(1 + \epsilon)^j$ for $0 \le j \le \lfloor \log_{1+\epsilon}{h} \rfloor = O(\frac{\log{h}}{\epsilon})$. As a result, we reduce the online single buyer auction problem to the multi-scale online learning with full information and finite actions. The action set has $k = O(\frac{\log{h}}{\epsilon})$ actions whose ranges form a geometric sequence $(1 + \epsilon)^j$, $0 \le j < k$.

Recall the definition of $\opt$ in Section~\ref{sec:multi-scale-nonnegative-results}, and let $\pbest$ be the best fixed price in hindsight, which is the price that achieves $\opt$. We now show how to get a multiplicative cum additive approximation for this problem with $\opt$ as the benchmark, \`a la \citet{blum2004online,blum2005near}.
The main improvement over these results is that the additive term scales with the best price rather than $h$. 
\begin{theorem}
	\label{thm:multadditivebounds-simple}
There is  an algorithm for the online single buyer auction problem that takes as input a parameter $\epsilon >0$, and satsify $\alg \geq (1-\epsilon) \opt - O(E)$, where:
	\[ 
	E = \frac{\pbest \log (\myfrac{\log h}{\epsilon})}{\epsilon}~.
	\]
	Also, even if $h$ is not known up front, there is an (slightly modified) algorithm that achieves a similar approximation guarantee for online single buyer auction with:
	\[ 
	E = \frac{\pbest \log (\myfrac{\pbest}{\epsilon})}{\epsilon}~.
	\]
\end{theorem}



\begin{proof}[\textbf{of Theorem~\ref{thm:multadditivebounds-simple}}]\\
\noindent\emph{\large{[Part 1: known $h$]}}~Recall the above formulation of the problem as an online learning problem with full information.
The proof then follows by Theorem~\ref{thm:expert-regret-bound-rewards}, letting $\pib$ to be the uniform distribution over the $k = O(\myfrac{\log{h}}{\epsilon})$ actions, i.e., discretized prices. 

\noindent\emph{\large{[Part 2: unknown $h$]}}~When $h$ is not known up front, we consider a variant of our algorithm (Algorithm~\ref{alg:single-auction-unknown}) that picks the next price in each round $t$  from the set of \emph{relevant prices} (denoted by $\mathcal{P}$), updates this set if necessary, and then updates the weights of prices in this set as in Algorithm~\ref{alg:MSMW}. The main new idea here is to update the set of prices $\mathcal{P}$ so that it only includes prices that are at most the highest value we have seen so far (let the highest seen value be $1$ at the beginning). Now, for the sake of analysis, consider a hypothetical algorithm (called $\texttt{ALG}^{H}$) that considers a countably infinite action space comprising all prices of the form $(1+\epsilon)^j$, for $j \ge 0$. We first show this hypothetical algorithm $\texttt{ALG}^{H}$ satisfies the required approximation guarantee in Theorem~\ref{thm:multadditivebounds-simple}. We then show the expected revenue of Algorithm~\ref{alg:single-auction-unknown} is at least the expected revenue of $\texttt{ALG}^{H}$ (minus a constant that is negligible in our bound), and hence the final proof.

The proof of the regret bound of Theorem~\ref{thm:expert-regret-bound-rewards} works when we have countably many actions (although we cannot implement such algorithms directly). Now, consider simulating $\texttt{ALG}^{H}$ and let the prior distribution $\pib$ be such that for any price $p = (1 + \epsilon)^j$, $\pi_p = \epsilon(\eps+2) (1 + \epsilon)^{-2(j+1)} = \frac{\epsilon(\eps+2)}{(1 + \epsilon)^2} \cdot \frac{1}{p^2}$ (this choice will become more clear later in the proof; in short we need $\pi_p$ to be proportional to $\tfrac{1}{p^2}$). The approximation guarantee in Theorem~\ref{thm:multadditivebounds-simple} then follows by Theorem~\ref{thm:expert-regret-bound-rewards}. We now argue the followings:
\begin{itemize}
\item For any round $t$, unless the value in that round is a new highest value, Algorithm~\ref{alg:single-auction-unknown} gets weakly higher revenue than $\texttt{ALG}^{H}$. This is because the probability that Algorithm~\ref{alg:single-auction-unknown} plays any relevant price in $\mathcal{P}$ (that has a non-zero gain in this round) is weakly higher than that in $\texttt{ALG}^{H}$.
\item For any price $p = (1+\eps)^j$, consider the first time a value at least $p$ shows up. Algorithm~\ref{alg:single-auction-unknown} suffers a loss of at most $p\cdot \pi_p$ compared to $\texttt{ALG}^{H}$, due to $\texttt{ALG}^{H}$'s probability of playing $p$ in that round, where $\pi_p$ is the probability of playing $p$ in the initial distribution. This is because the probability that $\texttt{ALG}^{H}$ plays $p$ in this round is at most $\pi_p$ as $p$ has not got any positive gains before this round.

\item Then, by choosing $\pi_p$ to be inversely proportional to $p^2$, we can show that Algorithm~\ref{alg:single-auction-unknown} has an additive loss of $\sum_{p} \tfrac{\beta}{p}=\tfrac{\eps+2}{\eps+1}=O(1)$ compared to$\texttt{ALG}^{H}$, where $\beta=\left(\sum_{p}\tfrac{1}{p^2}\right)^{-1}=\tfrac{\eps(2+\eps)}{(1+\eps)^2}$ is the normalization constant of the initial distribution $\pmb{\pi}$. This finishes the proof.
\end{itemize} 
\end{proof}
%
%
%

\begin{algorithm}[h]
\caption{Online single buyer auction (for unknown $h$)}
\label{alg:single-auction-unknown}
\begin{algorithmic}[1]
	\State{\textbf{input}} learning rate $0 < \eta \le 1$, price discretization parameter $0<\epsilon\le 1$.
  \State{\textbf{initialize}}~ the set of relevant prices $\mathcal{P}=\{1\}$. Let $\alpha_1(1)=1$.
	\For{$t=1,\dots,T$}
		\State Randomly pick a price in $\mathcal{P}$ drawn from $\pmb{\alpha}(t)$, and observe $\rewb{t}$. 
		\State Update $\mathcal{P}$ to be all the prices $(1+\epsilon)^j$ that are at most the highest value until time $t$.
		\State $\forall p \in \mathcal{P}:~~\w{p}{t+1}\leftarrow\alpha_{p}({t})\cdot \exp(\eta\cdot\frac{\rew{p}{t}}{p})$. 
		\State Find $\lambda^*$ (e.g., binary search) s.t. $\sum_{p \in \mathcal{P}}\w{p}{t+1}\cdot\exp(-\frac{\lambda^*}{p})=1$.
		\State $\forall p \in \mathcal{P}:~~\alpha_{p}({t+1})\leftarrow \w{p}{t+1}\cdot\exp(-\frac{\lambda^*}{p}).$
  \EndFor		
\end{algorithmic}
\end{algorithm}

Bounds on the sample complexity of auctions for single buyer problem~\citep{HuangMR15} imply that the first bound in this theorem is tight up to $\log$ factors: the lower bound is $h\epsilon^{-1}$ in an instance where $\pbest$ is actually equal to $h$.  Also, the best upper bound known is by \citet{blum2004online,blum2005near}, which is 
\[E=\frac{h \log (\myfrac{1}{\epsilon})}{\epsilon}~.\]
We conclude that Theorem~\ref{thm:multadditivebounds-simple} generalizes the known tight sample complexity upper-bound for the offline single buyer Bayesian revenue maximization to the online adversarial setting.

%% file: OMD-NegEntropy.tex
While it is possible to analyze the regret of the MSMW algorithm (Algorithm~\ref{alg:MSMW}) by using first principles, we take a different approach (the elementary analysis can still be found in the appendix, Section~\ref{appendix:OMD-first-principles}). We show how this algorithm is indeed an instance of the Online Mirror Descent (OMD) algorithm for a particular choice of the \emph{Legendre function} (also known as \emph{the mirror map}). 
\subsubsection{Preliminaries on online mirror descent.} Fix an open convex set $\mathcal{D}$ and its closure $\bar{\mathcal{D}}$, which in our case are $(0,+\infty)^{\expset}$ and $[0,+\infty)^{\expset}$ respectively, and a closed-convex action set $\mathcal{A}\subset \bar{\mathcal{D}}$, which in our case is $\Delta_\expset$, i.e. the set of all probability distributions over experts in $\expset$. At the heart of an OMD algorithm there is a \emph{Legendre} function $F:\bar{\mathcal{D}}\rightarrow \mathbb{R}$, i.e. a strictly convex function that admits continuous first order partial derivatives on $\mathcal{D}$ and $\lim_{x\to\bar{\mathcal{D}}\setminus\mathcal{D}}\lVert \nabla F(x)\rVert=+\infty$, where $\nabla F(.)$ denotes the gradient map of $F$. One can think of OMD as a member of \emph{projected gradient descent} algorithms, where the \emph{gradient update} happens in the \emph{dual space} $\nabla F(\mathcal{D})$ rather than in primal $\mathcal{D}$, and the \emph{projection} is defined by using the \emph{Bregman divergence} associated with $F$ rather than $\ell_2$-distance (see Figure~\ref{fig:omd}).

\begin{figure}[htb]
\begin{center}
\label{fig:omd}
\includegraphics[scale=0.5]{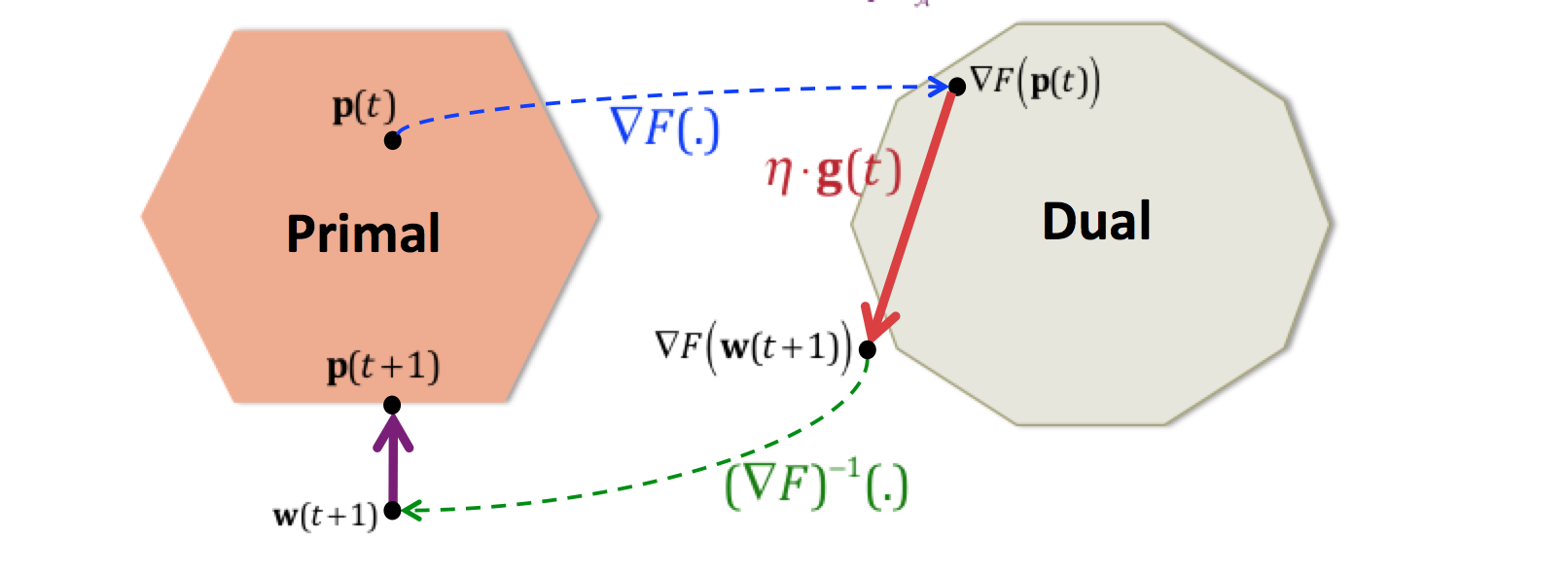}
\caption{Online Mirror Descent (OMD): moving to the dual space by gradient map \emph{(blue)}, gradient update in the dual space \emph{(red)}, applying the inverse gradient map \emph{(green)}, and finally projecting back to the simplex using Bregman projection \emph{(purple)}.}
\end{center}
\end{figure}

\begin{definition}[Bregman Divergence~\citep{bubeck2011introduction}]
\label{def:BD}
 Given a Legendre function $F$ over $\Delta_{\expset}$, the Bregman divergence associated with $F$, denoted as $D_F:\Delta_{\expset}\times\Delta_{\expset}\to\mathbb{R}$, is defined by
\begin{equation*}
D_F(x,y)=F(x)-F(y)-(x-y)^T\nabla F(y)
\end{equation*}
\end{definition}
 
 \begin{definition}[Online Mirror Descent~\citep{bubeck2011introduction}] Suppose $F$ is a Legendre function. At every time $t\in[T]$, the online mirror descent algorithm with Legendre function $F$ selects an expert drawn from distribution $\pb{t}$, and then updates $\wb{t}$ and $\pb{t}$ given rewards $\rewb{t}$  by:
\begin{align}
& \textrm{\emph{Gradient update:}} \notag \\
& \quad\quad 
\label{eq:gradient-update}
\nabla F(\wb{t+1})=\nabla F(\pb{t})+\eta\cdot \rewb{t}\Rightarrow \wb{t+1}=(\nabla F)^{-1}\left(\nabla F(\pb{t})+\eta\cdot \rewb{t}\right)\\
& \textrm{\emph{Bregman projection:}} \notag \\
& \quad\quad 
\label{eq:projection}
\pb{t+1}=\argmin{\mathbf{p}\in\Delta_\expset}{D_F(\mathbf{p},\wb{t+1})}
\end{align}
where $\eta>0$ is called the learning rate of OMD.
 \end{definition}
 
We use the following standard regret bound of OMD (Refer to~\cite{bubeck2011introduction} for a thorough discussion on OMD. For completeness, a proof is also provided in the appendix, Section~\ref{appendix:OMD-proof}).  Roughly speaking, this lemma upper-bounds the regret by the summation of two separate terms: ``local norm'' (the first term), which captures the total deviation between $\pb{t}$ and $\wb{t+1}$, and ``initial divergence'' (the second term), which captures how much the initial distribution is far from the target distribution.
\begin{lemma} 
\label{lem:OMD-regret}
For any learning rate parameter $0<\eta\leq 1$ and any benchmark distribution $\qb$ over $\actions$, the OMD algorithm with Legendre function $F(.)$ admits the following:
\begin{equation}\textstyle
\sum_{t\in[T]} \rewb{t} \cdot \big( \qb - \pb{t} \big)\leq \frac{1}{\eta}\sum_{t\in[T]}D_F(\pb{t},\wb{t+1})+\frac{1}{\eta}D_F(\qb,\pb{1})
\end{equation} 
\end{lemma}
 \subsubsection{MSMW algorithm as an OMD}
For our application, we focus on a particular choice of Legendre function that captures different learning rates proportional to $\range_i^{-1}$ for different experts, as we saw earlier in Algorithm~\ref{alg:MSMW}. We start by defining the \emph{weighted negative entropy} function.
\begin{definition}
Given expert-ranges $\{\range_i\}_{i\in\expset}$, the \emph{weighted negative entropy} is defined by
\begin{equation}\textstyle
F(x)=\sum_{i\in\expset} \range_i\cdot x_i\ln(x_i)
\end{equation}
\end{definition}
\begin{corollary}
\label{cor:neg-entropy}
 It is straightforward to see $F(x)=\sum_{i\in\expset} \range_i\cdot x_i\ln(x_i)$ is a non-negative Legendre function over $\mathbb{R}_+^{\expset}$. Moreover, $\nabla F(x)_i=\range_i(1+\ln(x_i))$ and $D_F(x,y)=\sum_{i\in\expset}\range_i\cdot (x_i\ln(\frac{x_i}{y_i})-x_i+y_i)$.
\end{corollary}
We now have the following lemma that shows Algorithm~\ref{alg:MSMW} is indeed an OMD algorithm.
\begin{lemma}
\label{lem:MSMWisOMD}
The MSMW algorithm, i.e. Algorithm~\ref{alg:MSMW}, is equivalent to an OMD algorithm associated with the weighted negative entropy $F(x)=\sum_{i\in\expset}\range_i\cdot x_i\ln(x_i)$ as its Legendre function.
\end{lemma} 

\begin{proof}
Look at the gradient update step of OMD, as  in Equation~\eqref{eq:gradient-update}, with Legendre function $F(x)=\sum_{i\in\expset}\range_i\cdot x_i\ln(x_i)$. By using Corollary~\ref{cor:neg-entropy} we have
\begin{align*}
 &\nabla F(\wb{t+1})=\nabla F(\pb{t})+\eta\cdot \rewb{t}\Rightarrow \range_i(1+\ln(\w{i}{t+1}))=\range_i(1+\ln(\p{i}{t}))+\eta\cdot \rew{i}{t}~,
\end{align*}
and therefore, $\w{i}{t+1}=\p{i}{t}\cdot \exp(\eta\cdot\frac{\rew{i}{t}}{\range_i})$. Moreover, for the Bregman projection step we have
\begin{equation}
 \pb{t+1}=\argmin{\mathbf{p}\in\Delta_\expset}{D_F(\mathbf{p},\wb{t+1})}=\argmin{\mathbf{p}\in\Delta_\expset}{\sum_{i\in\expset}{\range_i\cdot(p_i\ln(\frac{p_i}{\w{i}{t+1}})-p_i+\w{i}{t+1})}}
\end{equation}
This is a convex minimization over a convex set. To find a closed form solution, we look at the Lagrangian dual function $\mathcal{L}(\mathbf{p},\lambda)\triangleq\sum_{i\in\expset}{\range_i\cdot(p_i\ln(\frac{p_i}{\w{i}{t+1}})-p_i+\w{i}{t+1})}+\lambda(\sum_{i\in\expset}p_i-1)$ and the Karush-Kuhn-Tucker (KKT) conditions $\nabla \mathcal{L}(\mathbf{p^*},\lambda^*)=\mathbf{0}$. We have
\begin{equation}
\range_i\cdot\ln(\frac{p^*_i}{\w{i}{t+1}})+\lambda^*=0\Rightarrow p^*_i=\w{i}{t+1}\cdot\exp({-\frac{\lambda^*}{\range_i}})
\end{equation}
As $\sum_{i\in\expset}p^*_i=1$, $\lambda^*$ should be unique number s.t. $\sum_{i\in\expset}\w{i}{t+1}\cdot\exp(-\frac{\lambda^*}{\range_i})=1$, and then $\p{i}{t+1}=\w{i}{t+1}\cdot\exp({-\frac{\lambda^*}{\range_i}})$. So, Algorithm~\ref{alg:MSMW} is equivalent to OMD with weighted negative entropy as its Legendre function. 
\end{proof}

By combining Lemma~\ref{lem:OMD-regret}, Corollary~\ref{cor:neg-entropy} and finally Lemma~\ref{lem:MSMWisOMD} we prove the following regret bound for the MSMW algorithm. We encourage the reader to also look at the appendix, Section~\ref{appendix:OMD-first-principles},  for an extra proof using first principles.

\begin{proposition}
\label{prop:expert-regret-bound-general}
For any initial distribution $\mub$ over $\actions$, and any learning rate parameter $0 < \eta \le 1$, and any benchmark distribution $\qb$ over $\expset$, the MSMW algorithm satisfies that:
\[
\sum_{i \in \expset} q_i \cdot \reward{i} - \ex{\alg} \leq \eta \sum_{t \in [T]} \sum_{i\in \actions}\p{i}{t} \frac{(\rew{i}{t})^2}{\range_i} + \frac{1}{\eta} \cdot \sum_{i \in \expset} \range_i \bigg( q_i \ln \big( \frac{q_i}{\mu_i} \big) - q_i + \mu_i \bigg) ~.
\]

\end{proposition}
\begin{proof}[\textbf{of Proposition~\ref{prop:expert-regret-bound-general}}]
	We have:
	
	\begin{equation}
	\sum_{i \in \expset} q_i \cdot \reward{i} - \ex{\alg} = \sum_{t\in[T]} \qb \cdot \rewb{t}-\sum_{t\in[T]}\pb{t}\cdot\rewb{t}=\sum_{t\in[T]} \rewb{t} \cdot \big( \qb - \pb{t} \big)
	\end{equation}
	By applying the regret bound of OMD (Lemma~\ref{lem:OMD-regret}) to upper-bound the RHS, we have 
	\begin{equation}
	\sum_{i \in \expset} q_i \cdot \reward{i} - \ex{\alg}\leq \frac{1}{\eta}\sum_{t\in[T]}D_F(\pb{t},\wb{t+1})+\frac{1}{\eta}D_F(\qb,\pb{1})
	\end{equation}
	To bound the first term in regret, a.k.a \emph{local norm}, we have: 
	\begin{align}
		D_F(\pb{t},\wb{t+1})&=\sum_{i\in\expset}\range_i\cdot(\p{i}{t}\ln(\frac{\p{i}{t}}{\w{i}{t+1}})-\p{i}{t}+\w{i}{t+1})\nonumber\\
		&\label{eq:local-norm}=\sum_{i\in\expset}\range_i\cdot\p{i}{t}(-\eta\cdot\frac{\rew{i}{t}}{\range_i}-1+exp(\eta\cdot\frac{\rew{i}{t}}{\range_i}))
	\end{align}
	Note that $\eta \cdot \frac{\rew{i}{t}}{\range_i} \in [-1, 1]$ because $\rew{i}{t} \in [-\range_i, \range_i]$ and $0 < \eta \le 1$.
	By $\exp(x)-x-1\leq x^2$ for $-1 \le x \le 1$ and that $\eta \rew{i}{t} \in [-\range_i, \range_i]$, the above is upper bounded by $\eta^2\sum_{i\in \actions}\p{i}{t} \frac{(\rew{i}{t})^2}{\range_i}$. We can also rewrite the second term in regret.  In fact, if we set $\pb{1}=\mub$, then
	\[ 
	\frac{1}{\eta}\cdot D_F(\qb,\pb{1})=\frac{1}{\eta} \cdot \sum_{i \in \expset} \range_i \bigg( q_i \ln \big( \frac{q_i}{\mu_i} \big) - q_i + \mu_i \bigg)
	\]
	By summing the upper-bounds $ \eta^2\sum_{i\in \actions}\p{i}{t} \frac{(\rew{i}{t})^2}{\range_i}$ on each term of local norm in (\ref{eq:local-norm}) for $t\in[T]$ and putting all the pieces together, we get the desired bound.
\end{proof}

%

%% file: bandit.tex
\section{Multi-Scale Online Learning with Bandit Feedback}
\label{sec:bandit}
In this section, we look at the bandit feedback version of multi-scale online learning framework proposed in Section~\ref{sec:multi-scale-nonnegative-results}. Essentially, the only difference here is that after the algorithm picks an arm $i_t$ at time $t$, it only observes the obtained reward, i.e. $\rew{i_t}{t}$, and \emph{does not} observe the entire reward function $\rewb{t}$.

Inspired by the online stochastic mirror descent algorithm~\citep{bubeck2011introduction} we introduce \emph{Bandit-MSMW} algorithm. Our algorithm follows the standard bandit route of using unbiased estimators for the rewards in a full information strategy (in this case MSMW). We also mix the MSMW distribution with an extra uniform exploration, and use a tailored initial distribution to obtain the desired mutli-scale regret bounds.

\subsection{Bandit multi-scale regret bounds}

For the bandit version, we can get similar regret guarantees as in Section~\ref{sec:full-info-regret-simple} for the full-information variant, but only for the {\em best} action. 
If we require the regret bound to hold for all actions, then we can only get a weaker bound, where the second term has $\epsilon^{-2}$ instead of $\epsilon^{-1}$. 
The  difference between the bounds for the bandit and the full information setting is essentially  a factor of $k$, which is unavoidable. 
\begin{theorem}	\label{thm:bandit-regret-bound}
There exists an algorithm for the online multi-scale problem with bandit feedback that takes as input the ranges  $c_i, ~\forall ~i \in A$,  and a parameter $0 < \epsilon \le 1$,  and satisfies, 
\begin{itemize}
	\item for $i^* = \arg \max_{i \in \actions} \reward{i}$, 
	\begin{equation}\textstyle
	\ex{\regret{i^*}}\leq \epsilon\cdot\reward {i^*}+O \left( \frac{1}{\epsilon} k \log \big( \frac{k}{\epsilon} \big) \cdot \range_{i^*} \right).
	\end{equation}
	\item  for all $i\in \actions$, 
	\begin{equation}\textstyle
	\ex{\regret{i}}\leq \epsilon\cdot\reward{i} + O \left( \frac{1}{\epsilon^2} k \log \big( \frac{k}{\epsilon} \big) \cdot \range_{i} \right).
	\end{equation}
\end{itemize}
\end{theorem}


%
 Also, one can compute the pure-additive versions of the bounds in Theorems~\ref{thm:bandit-regret-bound}  by setting $\epsilon=\sqrt{\frac{k\log(kT)}{T}}$ and $\epsilon=(\frac{k\log(kT)}{T})^{\frac{1}{3}}$ resepctively (Corollary~\ref{cor:bandit-regret-bound-rewards-additive}), and compare with the pure-additive regret bound $O\big( \range_{\max}\cdot \sqrt{Tk}\big)$ for the adversarial multi-armed bandit problem~\citep{audibert2009minimax,auer1995gambling}.
\begin{corollary}
	\label{cor:bandit-regret-bound-rewards-additive}
	There exist algorithms for the online multi-scale bandits problem that satisfies, 
	\begin{itemize}
\item For $i^* = \arg \max_{i \in \actions} \reward{i}$,
	\begin{equation}
	\ex{\regret{i^*}} \leq {O} \left( \range_{i^*}\cdot\sqrt{Tk\log(kT)} \right)
	\end{equation}
\item For all $i\in\actions$,
	\begin{equation}
	\ex{\regret{i}} \leq {O} \left( \range_{i}\cdot T^{\frac{2}{3}}(k\log(kT))^{\frac{1}{3}} \right)
	\end{equation}
		\end{itemize}
\end{corollary}

Here is a map of this section. In Section~\ref{sec:bandit-MSMW} we propose our bandit algorithm and prove its general regret guarantee for non-negative rewards. Then in Section~\ref{sec:bandit-regret} we show how to get a multi-scale style regret guarantee for the best arm $\range_{i^*}$, and a weaker guarantee for all arms $\{\range_i\}_{i\in\actions}$.
\subsection{Bandit Multi-Scale Multiplicative Weight (Bandit-MSMW) algorithm}
\label{sec:bandit-MSMW}
We present our Bandit algorithm (Algorithm~\ref{alg:Bandit-MSMW}) when the set of actions $\actions$ is finite (with $|A| = k$). Let $\eta$ be the learning rate and $\gamma$ be the exploration probability. We show the following regret bound. 

\begin{algorithm}[htb]
\caption{Bandit-MSMW}
\label{alg:Bandit-MSMW}
\begin{algorithmic}[1]
	\State{\textbf{input}}~ exploration parameter $\gamma>0$, learning rate $\eta>0$.
  \State{\textbf{initialize}}~ $\pb{1} = (1 - \gamma) \mathbf{1}_{\imin} + \frac{\gamma}{k} \mathbf{1}$, where $\imin$ is the arm with minimum range $\range_{\imin}$.  
	\For{$t=1,\dots,T$}
		\State Let $\tpb{t} = (1 - \gamma) \pb{t} + \frac{\gamma}{k} \mathbf{1}$.
		\State Randomly pick an expert $i_t$ drawn from $\tpb{t}$, and observe $\rew{i_t}{t}$. 
		\State Let $\trewb{t}$ be such that
		\begin{align*}
			\trew{i}{t} = \begin{cases}
				\tfrac{\rew{i}{t}}{\tp{i}{t}} & \text{if $i = i_t$}; \\[1ex]
				0 & \text{otherwise}.
			\end{cases}
		\end{align*}
		\State $\forall i \in \actions:~~\w{i}{t+1}\leftarrow\p{i}{t}\cdot \exp(\frac{\eta}{\range_i} \cdot \trew{i}{t})$. 
		\State Find $\lambda^*$ (e.g., binary search) s.t. $\sum_{i \in \actions}\w{i}{t+1}\cdot\exp(-\frac{\lambda^*}{\range_i})=1$.
		\State $\forall i \in \actions:~~\p{i}{t+1}\leftarrow \w{i}{t+1}\cdot\exp(-\frac{\lambda^*}{\range_i}).$
  \EndFor		
\end{algorithmic}
\end{algorithm}

\begin{lemma}
\label{lem:bandit-msmw}
For any exploration probability $0<\gamma\leq \frac{1}{2}$ and any learning rate parameter $0 < \eta \le \frac{\gamma}{k}$, the Bandit-MSMW algorithm achieves the following regret bound when the gains are non-negative :
%
\[\textstyle
\forall i\in\actions: \ex{\regret{i}}\leq O \left( \frac{1}{\eta} \log \big( \frac{k}{\gamma} \big) \cdot \range_i + 
\eta \sum_{j \in \actions} \reward{j} + \gamma \cdot \reward{i} \right)
\]
\end{lemma}
\begin{proof}[\textbf{of Lemma~\ref{lem:bandit-msmw}}]
\label{appendix:bandit-regret-proof}
We further define:
\begin{eqnarray*}
\talg & \triangleq & \textstyle \sum_{t\in[T]} \rew{i_t}{t} = \sum_{t\in[T]}\tpb{t}\cdot\trewb{t}~,\\
\treward{j} & \triangleq & \textstyle \sum_{t\in[T]}\trew{j}{t}~.
\end{eqnarray*}
%

In expectation over the randomness of the algorithm, we have:
\begin{enumerate}
	\item $\ex{\alg} = \ex{\talg}$; and
	\item $\reward{j} = \ex{\treward{j}}$ for any $j \in \actions$.
\end{enumerate}
%

Hence, to upper bound $\ex{\regret{i}}=\reward{i}-\ex{\alg}$, it suffices to upper bound $\ex{\treward{i} - \talg}$.

By the definition of the probability that the algorithm picks each arm, i.e., $\tpb{t}$, we have:
\[\textstyle
\ex{\talg} \ge (1 - \gamma) \sum_{t\in[T]} \pb{t} \cdot \trewb{t} ~.
\]

Hence, we have that for any initial distribution $\qb$ over $\actions$:
\begin{align}
\sum_{j \in \actions} q_j \cdot \ex{\treward{j}} - \ex{\talg} \le ~ & \textstyle \ex{\sum_{j \in \actions} q_j \cdot \treward{j} - \sum_{t\in[T]} \pb{t} \cdot \trewb{t}} + \frac{\gamma}{1-\gamma} \ex{\talg} \notag \\
\le ~ & \textstyle \ex{\sum_{j \in \actions} q_j \cdot \treward{j} - \sum_{t\in[T]} \pb{t} \cdot \trewb{t}} + 2 \gamma \ex{\talg} ~. \label{eqn:bandit-msmw1}
\end{align}

Next, we upper bound the 1st term on the RHS.
Note that $\pb{t}$'s are the probabilities of choosing experts by MSMW when the experts have rewards $\trewb{t}$'s.
By Proposition~\ref{prop:expert-regret-bound-general}, we have that for any benchmark distribution $\qb$ over $S$, the Bandit-MSMW algorithm satisfies that:
\begin{equation}
\label{eqn:bandit-msmw2}
\sum_{j \in \actions} q_j \cdot \treward{j} - \sum_{t\in[T]} \pb{t} \cdot \trewb{t} \leq \eta \sum_{t\in[T]} \sum_{j \in \actions} \frac{\p{j}{t}}{\range_j} \cdot \big( \trew{j}{t} \big)^2 + \frac{1}{\eta} \sum_{j \in \actions} \range_j \left( q_j \ln \big( \frac{q_j}{\p{j}{1}} \big) - q_j + \p{j}{1} \right) ~.
\end{equation}

For any $t \in [T]$ and any $j \in \actions$, by the definition of $\trew{j}{t}$, it equals $\tfrac{\rew{j}{t}}{\tp{j}{t}}$ with probability $\tp{j}{t}$, and equals $0$ otherwise.
Thus, if we fix the random coin flips in the first $t-1$ rounds and, thus, fix $\tpb{t}$, and take expectation over the randomness in round $t$, we have that:
\[
\ex{\frac{\p{j}{t}}{\range_j} \cdot \big( \trew{j}{t} \big)^2} = \frac{\p{j}{t}}{\range_j} \cdot \tp{j}{t} \cdot \left(\frac{\rew{j}{t}}{\tp{j}{t}} \right)^2 = \frac{\p{j}{t}}{\tp{j}{t}} \frac{(\rew{j}{t})^2}{\range_j} ~.
\]

Further note that since $\tp{j}{t} \ge (1 - \gamma) \p{j}{t}$, and $\rew{j}{t} \le \range_j$, the above is upper bounded by $\frac{1}{1 - \gamma} \rew{j}{t} \le 2 \rew{j}{t}$.
Putting together with \eqref{eqn:bandit-msmw2}, we have that for any $0 < \eta \le \frac{\gamma}{n}$:
\begin{align*}
\ex{\sum_{j \in \actions} q_j \cdot \treward{j} - \sum_{t\in[T]} \pb{t} \cdot \trewb{t}}
\le ~
&
\eta \sum_{t\in[T]}\sum_{j \in \actions} 2 \rew{j}{t} + \frac{1}{\eta} \sum_{j \in \actions} \range_j \left( q_j \ln \big( \frac{q_j}{\p{j}{1}} \big) - q_j + \p{j}{1} \right) \\
= ~
&
2 \eta \sum_{j \in \actions} \reward{j} + \frac{1}{\eta} \sum_{j \in \actions} \range_j \left( q_j \ln \big( \frac{q_j}{\p{j}{1}} \big) - q_j + \p{j}{1} \right)
\end{align*}

Combining with \eqref{eqn:bandit-msmw1}, we have:
\[
\sum_{j \in \actions} q_j \cdot \ex{\treward{j}} - \ex{\talg} \le 2 \eta \sum_{j \in \actions} \reward{j} + \frac{1}{\eta} \sum_{j \in \actions} \range_j \left( q_j \ln \big( \frac{q_j}{\p{j}{1}} \big) - q_j + \p{j}{1} \right) + 2 \gamma \ex{\talg}
\]

Let $\qb = (1 - \gamma) \mathbf{1}_{i} + \frac{\gamma}{k} \mathbf{1}$.
Recall that $\pb{1} = (1 - \gamma) \mathbf{1}_{\imin} + \frac{\gamma}{k} \mathbf{1}$ (recall $\imin$ is the arm with minimum range $\range_{\imin}$).
Similar to the discussion for the expert problem in Section~\ref{sec:expert-reward-only}, the 2nd term on the RHS is upper bounded by $O \left( \frac{1}{\eta} \log \big( \frac{k}{\gamma} \big) \cdot \range_i \right)$.
Hence, we have:
\begin{equation}
\label{eqn:bandit-msmw3}
\sum_{j \in \actions} q_j \cdot \ex{\treward{j}} - \ex{\talg} \le 2 \eta \sum_{j \in \actions} \reward{j} + O \left( \frac{1}{\eta} \log \big( \frac{k}{\gamma} \big) \cdot \range_i \right) + 2 \gamma \ex{\talg} ~.
\end{equation}

Further, the LHS is lower bounded as:
\[ 
(1 - \gamma) \ex{\treward{i}} + \frac{\gamma}{k} \sum_{j \in \actions} \ex{\treward{j}} - \ex{\talg} \ge (1 - \gamma) \ex{\treward{i}} - \ex{\talg} ~.
\]

The lemma then follows by putting it back to \eqref{eqn:bandit-msmw3} and rearranging terms.
\end{proof}


\subsection{Regret bounds for non-negative rewards - proof of Theorem~\ref{thm:bandit-regret-bound} }
\label{sec:bandit-regret}
\begin{proof}[\textbf{of Theorem~\ref{thm:bandit-regret-bound}}]
Letting $\gamma = \epsilon$ and $\eta = \frac{\gamma}{k} = \frac{\epsilon}{k}$ in Lemma~\ref{lem:bandit-msmw}, we get that the expected regret w.r.t.\ an action $i \in \actions$ is bounded by:
\[\textstyle
O \left( \epsilon \cdot \reward{i} + \frac{\epsilon}{k} \sum_{j \in \actions} \reward{j} + \range_i \cdot \frac{k}{\epsilon} \ln \big( \frac{k}{\epsilon} \big) \right) ~.
\]

When $i = \ibest$ (best arm), regret is bounded by $O \left( \epsilon \cdot \reward{\ibest} + \range_\ibest \cdot \frac{k}{\epsilon} \ln \big( \frac{k}{\epsilon} \big) \right)$, as desired.

For the regret w.r.t. an arbitrary action, 
note that $\ex{\alg} \ge \frac{\gamma}{k} \sum_{j \in \actions} \reward{j}$.
Thus, the regret bound w.r.t.\ an action $i \in \actions$ in Lemma~\ref{lem:bandit-msmw} is further upper bounded by:
\[\textstyle
O \left( \frac{1}{\eta} \log \big( \frac{k}{\gamma} \big)
\cdot \range_i + \left( \frac{\eta k}{\gamma} + \gamma\right) \cdot \ex{\talg} \right) 
\]

The theorem then follows by letting $\gamma = \epsilon$ and $\eta = \frac{\gamma^2}{k} = \frac{\epsilon^2}{k}$. 
%
\end{proof}


%% file: auction.tex
\section{More Applications of Multi-scale Learning for  Auctions and Pricing}
\label{sec:auction}
In this section, we consider applying the multi-scale online learning framework, developed in Section~\ref{sec:learning} and Section~\ref{sec:bandit}, to design several other online auctions and pricings be the single buyer auction (discussed in Section~\ref{sec:canonical-pricing}). Besides the single buyer auction, the problems that we consider are as follows. 
\begin{itemize}[itemindent=-2mm]
	\item\textbf{Online posted pricing:} 
		The same as the online single buyer auction of Section~\ref{sec:canonical-pricing}, but in the bandit setting. The algorithm only learns the indicator function $ \mathbf{1}(v(t) \geq p_t)$ where $p_t$ is the price it picks in round $t$.
	\item\textbf{Online multi buyer auction:}
		The action set is the set of all ``Myerson-type'' mechanisms for $n$ buyers, for some $ n\in \naturals$. (See Definition \ref{def:myerson-type-auction}.)
		The adversary picks a valuation vector $\vals t\in [1,h]^n$ and the reward of a mechanism $M$ is its revenue when the valuation of the buyers is given by $\vals t$; 
		this is denoted by $\rev_M (\vals t)$.
		 The algorithm sees the full vector of valuations $\vals t$. 		
\end{itemize}
\subsection{Auctions and pricing as multi-scale online learning problems}
\label{sec:auctions-as-learning}
We now show how to reduce the above problems to special cases of multi-scale online learning.

\paragraph{Online multi buyer auction}
In multi buyer auctions, we consider the set of all discretized Myerson-type auctions as the action space. 
We start by defining Myerson-type auctions:

\begin{definition}[Myerson-type auctions]
\label{def:myerson-type-auction}
A \emph{Myerson-type auction} is defined by $n$ non-decreasing virtual value mappings $\phi_1, \dots, \phi_n : [1,h] \mapsto [-\infty,h]$.
Given a value profile $v_1, \dots, v_n$, the item is given to the bidder $j$ with the largest non-negative virtual value $\phi_j(v_j)$.
Then, bidder $j$ pays the minimum value that would keep him as the the winner.
\end{definition}

\citet{Myerson} shows that when the bidders' values are drawn from independent (but not necessarily identical) distributions, the revenue-optimal auction is a Myerson-type auction.
\citet[Lemma~5]{devanur2016sample} observe that to obtain a $1-\epsilon$ approximation, it suffices to consider the set of discretized Myerson-type auctions that treat each bidder's value as if it is equal to the closest power of $1+\epsilon$ from below.
As a result, it suffices to consider the set of discretized Myerson-type auctions, each of which is defined by the virtual values of $(1+\epsilon)^j$'s, i.e., by $O ( \myfrac{n \log{h}}{\epsilon} )$ real numbers $\phi_\ell ( (1+\epsilon)^j )$, for $\ell \in [n]$, and $0 \le j \le \lfloor \log_{1+\epsilon}{h} \rfloor$. Furthermore, first \citet{elkind2007des} and later on
\citet{devanur2016sample,gonczarowski2016efficient} note that a discretized Myerson-type auction is in fact completely characterized by the total ordering of $\phi_\ell ( (1+\epsilon)^j )$'s;\footnote{\cite{cai2012optimal} also generalizes this observation to multi-dimensional types.} their actual values do not matter.
Indeed, both the allocation rule and the payment rule are determined by the ordering of virtual values.
As a result, our action space is a finite set with at most $O ( (\myfrac{n \log{h}}{\epsilon})! )$ actions.
The range of an action, i.e., a discretized Myerson-type auction, is the largest price ever charged by the auction, i.e., the largest value $v$ of the form $(1+\epsilon)^j$ such that there exists $\ell \in [n]$, $\phi_\ell(v) > \phi_\ell ( (1+\epsilon)^{-1} v )$.

\subsection{Multiplicative/additive approximations}

Similar to Section~\ref{sec:canonical-pricing}, we show how to get a multiplicative cum additive approximations for these problems with $\opt$ as the benchmark. Recall the definition of $\opt$ in Section~\ref{sec:multi-scale-nonnegative-results} and let $\pbest$ be the best fixed price on hindsight, which is the price that achieves $\opt$. 
\begin{theorem}
	\label{thm:multadditivebounds}
There are algorithms for the online posted pricing and the online multi buyer auction problems that take as input a parameter	$\epsilon >0$, and satsify $\alg \geq (1-\epsilon) \opt - O(E)$, where respectively (for the two problems mentioned above)
	\[ 
	E = \frac{\pbest \log h \log (\myfrac{\log h}{\epsilon})}{\epsilon^2}, 
	~~~~\text{and }~~~~ \frac{h n \log{h} \log( \myfrac{n \log{h}}{\epsilon} )}{\epsilon^2} ~.
	\]
	Even if $h$ is not known up front, we can still get the similar approximation guarantee for the online multi buyer auction with:
	\[ 
	E = \frac{h n \log{h} \log( \myfrac{n \log{h}}{\epsilon} )}{\epsilon^2} ~.
	\]
\end{theorem}

We conjecture that our bound for the online posted pricing problem is tight up to logarithmic factors, and leave resolving this as an open problem. 
The second bound is not comparable to the best sample complexity for the multi buyer auction problem by \citet{roughgarden2016ironing}; it is better than theirs for large $\epsilon$ (when ${1}/{\epsilon} \leq o(nh)$), and is worse for smaller $\epsilon$ (when $1/\epsilon \geq \omega(nh)$).
Also, compare the first bound to the corresponding upper bound for the pricing problem by \citet{blum2005near}, which is
\[\min\left\{\frac{h \log h \log\log h}{\epsilon^2}, \frac{h  \log\log h}{\epsilon^3} \right\}. \]
Essentially, the main improvement over this result is that the additive term scales with the best price rather than $h$.

\subsection{Proof of Theorem~\ref{thm:multadditivebounds}}

\begin{proof}
\emph{Online posted pricing.~}
Recall the formulation of the problem as an online learning problem with bandit feedback in Section~\ref{sec:auctions-as-learning}. This part then follows by Theorem~\ref{thm:bandit-regret-bound} with $k = O(\myfrac{\log{h}}{\epsilon})$ actions.
\medskip

\emph{Online multi buyer auction.~}
Recall the formulation of the problem as an online learning problem with full information in Section~\ref{sec:auctions-as-learning}.
The proof then follows by Theorem~\ref{thm:expert-regret-bound-rewards}, where we let $\pib$ be the uniform distribution over the $k = O((\myfrac{n \log{h}}{\epsilon})!)$ actions, i.e., Myerson-type auctions. 

When $h$ is not known up front, similar to the proof of Theorem~\ref{thm:multadditivebounds-simple}, we consider a hypothetical  algorithm with countably infinite action space $\actions$ as follows.
For any $p = (1 + \epsilon)^j$, $j \ge 0$, let the $k_p = O((\myfrac{n \log{p}}{\epsilon})!)$ Myerson-type auctions for values in $[1,p]$ be in $\actions$; we assume these auctions treat any values greater than $p$ as if they were $p$.
Further, we choose the prior distribution $\pib$ such that the probability mass of each auction for range $[1,p]$ is equal to $\frac{\epsilon(\epsilon+2)}{(1 + \epsilon)^2} \cdot \frac{1}{p^2} \cdot \frac{1}{k_p}$.
The approximation guarantee then follows by Theorem~\ref{thm:expert-regret-bound-rewards}. To implement this algorithm, we use the same trick as in the proof Theorem~\ref{thm:multadditivebounds-simple} by running a modified algorithm that only considers auctions for all ranges $[1,p]$ where $p$ is no larger than the highest value seen so far among all the buyers (i.e. a multi-buyer auction version of Algorithm~\ref{alg:single-auction-unknown}). The rest of the proof that shows the revenue loss of this algorithm compared to the hypothetical algorithm is negligible is similar to the proof of Theorem~\ref{thm:multadditivebounds-simple} (and hence omitted for brevity). 

\end{proof}
\vspace{-4mm}
\subsection{Competing with \texorpdfstring{$\delta$-guarded benchmarks}{Lg}}
For the single buyer auction/pricing problem, we define a 
$\delta$-guarded benchmark, for any $\delta \in [0,1]$.
This benchmark is restricted to only those prices that sell 
the item in at least a $\delta$ fraction of the rounds. 
\[ \textstyle \opt(\delta)    \defeq  \max\left\{ \sum_{t=1}^T \rew{p}{t} :p \in \actions , \sum_{t=1}^T \mathbf{1}(v_t \geq p) \geq \delta T\right\}  .\]
As observed in Footnote~\ref{fn1}, one can replace $\delta$ with $1/h$ and get the corresponding guarantees for $\opt$ rather than $\optdelta$. 
However, the main point of these results is to show a graceful improvement of the bounds as $\delta$ is chosen to be larger. 
\paragraph{Multiple buyers:~}
For the multi buyer auction problem, we define the $\delta$-guarded benchmark as follows. 
For any sequence of value vectors $\vals 1,\vals 2, \ldots, \vals T$, 
let $\bar{V}$ denote the largest value such that there are at least $\delta T$ distinct $t\in[1:T]$ with 
$\max_{i \in [n]}  \val i t \ge \bar{V}$.
Define the $\delta$-guarded benchmark to be
$$\textstyle \optdelta = \max_M \sum_{t=1}^T Rev_M\left( \min(\bar{V} \one, \vals t))\right) ,$$
where the ``$\min$'' is taken coordinate-wise, 
and the ``max'' is over all  Myerson-type mechanisms. In other words, here is how we can describe the $\delta$-guarded benchmark: for each Myerson-type auction $M$, after identifying the value cap $\bar{V}$, we cut all the values that are above $\bar{V}$ by this quantity, and then run $M$. The benchmark is then the revenue of the best Myerson-type auction under these modified values.

We focus on purely multiplicative approximation factors when competing with $\optdelta $. In particular, for any given $\epsilon>0$, we are interested in a $1-\epsilon$ approximation. 
We state our results in terms of the \emph{convergence rate}. 
We say that $T(\epsilon,\delta)$ is the convergence rate of an algorithm if  for all time horizon $T \geq T(\epsilon, \delta)$, we are guaranteed that $\alg \geq (1-\epsilon)\optdelta$. 
Our main results are as follows. 
\begin{theorem}\label{thm:convergencerates}
	There are algorithms for the online single buyer auction, online posted pricing, and the online multi buyer auction problems with  convergence  rates respectively of $$O\left(\frac{\log (\myfrac{\log h}{\epsilon})}{\epsilon^2 \delta} \right), ~~~~
	O\left(\frac{\log h}{\epsilon^4\delta}\right), 
	~~~~\text{and } O\left( \frac{n \log{(\myfrac{1}{\epsilon \delta})} \log( \myfrac{n \log(\myfrac{1}{\epsilon \delta})}{\epsilon} )}{\epsilon^3 \delta} + \frac{\log{ (\myfrac{\log{h}}{\epsilon}) }}{\epsilon^2 \delta} \right).$$
	Even if $h$ is not known upfront, we can still get the following similar convergence rates for online single buyer auction and online multi buyer auction respectively:
	\[ 
	O\left(\frac{\log (\myfrac{\pbest}{\epsilon})}{\epsilon^2 \delta} \right), 
	~~~~\text{and }~~ O\left( \frac{n \log{(\myfrac{1}{\epsilon \delta})} \log( \myfrac{n \log(\myfrac{1}{\epsilon \delta})}{\epsilon} )}{\epsilon^3 \delta} + \frac{\log{ (\myfrac{h}{\epsilon}) }}{\epsilon^2 \delta} \right) ~.
	\]
\end{theorem}
	Once again, we compare to the sample compexity bounds:
	our first is within a $\log \log h$ factor of the best sample complexity upper bound in \citet{huang2015making}.
	The lower bound for the online single buyer auction is $\Omega(\delta^{-1}\epsilon^{-2})$,
	which is also the best lower bound known for the pricing and the multi-buyer problem.\footnote{ \citet{ColeR14} show that at least a linear dependence on $n$ is necessary when the values are drawn from a regular distribution, but as is, their lower bound needs unbounded valuations. The lower bound probably holds for ``large enough $h$'' but it is not clear if it holds for all $h$.}  
	For the online posted pricing problem, we conjecture that the right dependence on $\epsilon$ should be $\epsilon^{-3}$.  
	No sample complexity bounds for the multi-buyer problem were known before; in fact we introduce the definition of a $\delta$-guarded benchmark for this problem.

\subsection{Proof of Theorem~\ref{thm:convergencerates}}


\begin{proof}
\emph{Online single buyer auction.~}
By Theorem~\ref{thm:expert-regret-bound-rewards}, letting $\pib$ be the uniform distribution over the $k = O(\myfrac{\log{h}}{\epsilon})$ actions, i.e., discretized prices, we have that for any price $p$ (recall that $\range_p = p$):
\[
\alg \ge (1 - \epsilon) \cdot \reward{p} - O \left( \tfrac{\log ( \myfrac{\log{h}}{\epsilon} )}{\epsilon} \cdot p \right) ~.
\]
For the $\delta$-guarded optimal price $\pbest$ (i.e., subject to selling in at least $\delta T$ rounds), we have $\reward{\pbest} \ge \delta T \cdot \pbest$.
Therefore, when $T \ge O\left(\myfrac{\log (\myfrac{\log h}{\epsilon})}{\epsilon^2 \delta} \right)$, the additive term of the above approximation guarantee is at most $\epsilon \cdot \reward{\pbest}$.
So the theorem holds.

The treatment for the case when $h$ is not known up front is essentially the same as in Theorem~\ref{thm:multadditivebounds} and Theorem~\ref{thm:multadditivebounds-simple}.
As a hypothetical algorithm useful for analysis, we consider an algorithm (similar to Algorithm~\ref{alg:MSMW}) with a countably infinite action space comprising all prices of the form $(1+\epsilon)^j$, for $j \ge 0$.
Then, let the prior distribution $\pib$ be such that for any price $p = (1 + \epsilon)^j$, $\pi_p = \epsilon(\eps+2) (1 + \epsilon)^{-2(j+1)} = \frac{\epsilon(eps+2))}{(1 + \epsilon)^2} \cdot \frac{1}{p^2}$. The rest of the proof and how to implement is the same as in the proof of Theorem~\ref{thm:multadditivebounds-simple} (i.e. Algorithm~\ref{alg:single-auction-unknown}).

\medskip

\emph{Online posted pricing.~}
Recall the above formulation of the problem as an online learning problem with bandit feedback.
By Theorem~\ref{thm:bandit-regret-bound} with $k = O(\myfrac{\log{h}}{\epsilon})$ actions, we have that for any price $p$:
\[
\alg \ge (1 - \epsilon) \cdot \reward{p} - O \left( \tfrac{\log{h} \log ( \myfrac{\log{h}}{\epsilon} )}{\epsilon^3} \cdot p \right) ~.
\]
Again, for the $\delta$-guarded optimal price $\pbest$ (i.e., subject to selling in at least $\delta T$ rounds), we have $\reward{\pbest} \ge \delta T \cdot \pbest$.
Therefore, when $T \ge O \left( \myfrac{\log{h} \log \big( \myfrac{\log{h}}{\epsilon} \big)}{\epsilon^4 \delta} \right)$, the additive term of the above approximation guarantee is at most $\epsilon \cdot \reward{\pbest}$.
So the theorem holds.

\medskip

\emph{Online multi buyer auction.~}
Suppose $\ibest$ is the $\delta$-guarded best Myerson-type auction.
Recall that $\bar{V}$ is the largest value such that there are at least $\delta T$ distinct $v(t)$'s with $\max_{\ell\in[n]} v_\ell(t) \ge \bar{V}$.
So we may assume without loss of generality that $\ibest$ does not distinguish values greater than $\bar{V}$.
Hence:
\begin{equation}
\label{eqn:multi-buyer-convergence1}
\range_{\ibest} \le \bar{V} ~.
\end{equation}

Further, note that running a second-price auction with anonymous reserve $\bar{V}$ is a Myerson-type auction (e.g., mapping values less than $\bar{V}$ to virtual value $-\infty$ and values greater than or equal to $\bar{V}$ to virtual value $\bar{V}$), and it gets revenue at least $\delta T \cdot \bar{V}$.
So we have that:
\begin{equation}
\label{eqn:multi-buyer-convergence2}
\reward{\pbest} \ge \delta T \cdot \bar{V} ~.
\end{equation}

Finally, the above implies that to obtain a $1 - \epsilon$ approximation, it suffices to consider prices that are at least $\epsilon \delta \bar{V}$.
Hence, it suffices to consider Myerson-type auctions that, for a given $\bar{V}$, do not distinguish among values greater than $\bar{V}$, and do not distinguish among values smaller than $\epsilon \delta \bar{V}$.
There are $O(\myfrac{\log{h}}{\epsilon})$ different values of $\bar{V}$.
Further, given $\bar{V}$, there are only $O(\myfrac{\log(\myfrac{1}{\epsilon \delta})}{\epsilon})$ distinct values to be considered and, thus, there are at most $O((\myfrac{n \log(\myfrac{1}{\epsilon \delta})}{\epsilon})!)$ distinct Myerson-type auctions of this kind.
Hence, the total number of distinct Myerson-type actions that we need to consider is at most:
\[
k = O \left( \frac{\log{h}}{\epsilon} \cdot \left( \frac{n \log(\myfrac{1}{\epsilon \delta})}{\epsilon} \right)! \right) ~.
\]

Letting $\pib$ be the uniform distribution over the $k$ actions in Theorem~\ref{thm:expert-regret-bound-rewards}, we have that (recall Eqn.~\eqref{eqn:multi-buyer-convergence1}):
\[
\alg \ge (1 - \epsilon) \cdot \reward{\ibest} - O\left( \frac{n \log{(\myfrac{1}{\epsilon \delta})} \log( \myfrac{n \log(\myfrac{1}{\epsilon \delta})}{\epsilon} )}{\epsilon^2} + \frac{\log{ (\myfrac{\log{h}}{\epsilon}) }}{\epsilon} \right) \cdot \bar{V}~.
\]

When $T \ge O\left( \frac{n \log{(\myfrac{1}{\epsilon \delta})} \log( \myfrac{n \log(\myfrac{1}{\epsilon \delta})}{\epsilon} )}{\epsilon^3 \delta} + \frac{\log{ (\myfrac{\log{h}}{\epsilon}) }}{\epsilon^2 \delta} \right)$, the additive term of the above approximation guarantee is at most $\epsilon \cdot \reward{\ibest}$ due to Eqn.~\eqref{eqn:multi-buyer-convergence2}.
So the theorem holds.

Again, the treatment for the case when $h$ is not known up front is similar to that in Theorem~\ref{thm:multadditivebounds}.
When $h$ is not known up front, we consider a hypothetical algorithm with a countably infinite action space $\actions$ as follows.
For any $\bar{V} = (1 + \epsilon)^j$, $j \ge 0$, let the $k^{\prime} = O((\myfrac{n \log(\myfrac{1}{\epsilon \delta})}{\epsilon})!)$ Myerson-type auctions that do not distinguish among values greater than $\bar{V}$, and do not distinguish among values smaller than $\epsilon \delta \bar{V}$ be in $\actions$.
Further, we choose the prior distribution $\pib$ such that the probability mass of each Myerson-type auction for a given $\bar{V}$ is equal to $\frac{\epsilon}{1 + \epsilon} \cdot \frac{1}{\bar{V}} \cdot \frac{1}{k^{\prime}}$.
The approximation guarantee then follows by Theorem~\ref{thm:expert-regret-bound-rewards} and essentially the same argument as the known $h$ case. Implementation is similar to the proof of Theorem~\ref{thm:multadditivebounds} and Theorem~\ref{thm:multadditivebounds-simple} (i.e. a multi-buyer auction version of Algorithm~\ref{alg:single-auction-unknown}). The rest of the proof that shows the revenue loss of this algorithm compared to the hypothetical algorithm is negligible is similar to the proof of Theorem~\ref{thm:multadditivebounds} (and hence omitted for brevity). 
\end{proof}
\vspace{-4mm}
\paragraph{Remark}
\citet{devanur2016sample} show that when the values are drawn from independent regular distributions, the $\epsilon$-guarded optimal price is a $1 - \epsilon$ approximation of the unguarded optimal price.
So our convergence rate for the online multi buyer auction problem in Theorem~\ref{thm:expert-regret-bound-rewards} implies a $\tilde{O}( n \epsilon^{-4})$ sample complexity modulo a mild $\log\log{h}$ dependency on the range, almost matching the best known sample complexity upper bound for regular distributions.

%% file: symmetric.tex
\section{Multi-scale Online Learning with Symmetric Range}
\label{sec:symmetric}
In this section, we consider multi-scale online learning when the rewards are in a symmetric range, i.e. for all $i\in\actions$ and $t\in[T]$, $\rew{i}{t}\in [-\range_i,\range_i]$. The standard analysis for the experts and the bandit problems holds even if the range of $\rew{i}{t}$ is $[-\range_i,\range_i],$ instead of  $[0,\range_i]$.  In contrast, there are subtle differences on the best achievable multi-scale regret bounds between the non-negative and the symmetric range, which we explore in this section. We look at both the full information and bandit setting, and prove action-specific regret upper bounds. We then prove a tight lower-bound in Section~\ref{sec:proof-log-dep} for the full information case, and an almost tight lower-bound in Section~\ref{sec:bandit-lower-bound-symmetric} for the bandit setting.

\subsection{Multi-scale regret bounds for symmetric ranges}
We first show the following upper bound for the full information setting when the range is symmetric. This bound follows the same style of action-specific regret bounds as in Theorem~\ref{thm:expert-regret-bound-rewards}. 
More detailed discussion on how the choice of initial distribution $\pib$ affects the bound is deferred to the appendix, Section~\ref{appendix:bandit-symmetric-initial} (recall that the initial distribution $\pib$ is the distribution over actions that is used in the first round of Algorithm~\ref{alg:MSMW}).
\begin{theorem}
	\label{thm:expert-regret-bound-symmetric}
There exists an algorithm for the multi-scale experts problem with symmetric range that	takes as input any distribution $\pib$ over $\expset$, the ranges $c_i, ~\forall ~i \in A$, and a parameter $0 < \epsilon \le 1$,  and satisfies:
	\begin{equation}
	\label{eqn:expert-regret-bound-symmetric}
	\textstyle
	\forall i \in \actions:~~\ex{\regret{i}}\leq \epsilon\cdot\ex{ \sum_{t\in[T]} \big| \rew{t}{i} \big|}+O \bigg( \frac{1}{\epsilon} \log \big( \frac{1}{\pi_i} \cdot \frac{\range_i}{\range_{\min}}  \big) \cdot \range_i \bigg) .
	\end{equation}
\end{theorem}

Similar to Section~\ref{sec:multi-scale-nonnegative-results}, we can compute the pure-additive version of the bound in Theorem~\ref{thm:expert-regret-bound-symmetric} by setting $\epsilon=\sqrt{\frac{\log (k\cdot \frac{\range_{\max}}{\range_{\min}})}{T}}$, as in Corollary~\ref{cor:expert-regret-bound-rewards-additive}. 
\begin{corollary}
	\label{cor:expert-regret-bound-rewards-additive-symmetric}
	There exists an algorithm for the online multi-scale experts problem with symmetric range that  takes as input the ranges $c_i, ~\forall ~i \in A$,  and satisfies:
	\begin{equation}\textstyle
	\forall i \in \actions:~~\ex{\regret{i}} \leq {O} \left( \range_{i}\cdot\sqrt{T\log (k\cdot \frac{\range_{\max}}{\range_{\min}})} \right)
	\end{equation}
\end{corollary}
If we compare the above regret bound with the standard $O(\range_{\max}\sqrt{T\log k})$ regret bound for the experts problem, we see that we replace the dependency on $\range_{\max}$ in the standard bound with $  \range_i \sqrt{\log(\frac{\range_{\max}}{\range_{\min}})}$.
It is natural to ask whether we could get rid of the dependence on $\log(\range_i/\range_{\min})$ and show a regret bound of  $O(\range_{i}\sqrt{T\log k})$, like we did for non-negative rewards.
However, the next theorem shows that this dependence on $\log(\range_i/\range_{\min})$ in the above bound is necessary, in a weak sense: where the constant in the $O(\cdot)$ is universal and does not depend on the ranges $c_i$. 
This is because the lower bound only holds for ``small'' values of the horizon $T$, which nonetheless grows with the $\{c_i\}$s.\footnote{For this reason we chose not to include this bound in Table \ref{tab:multi-scale}.}
\begin{theorem}
	\label{thm:expert-log-dependency-symmetric}
	There exists an action set of size $k$, and ranges $c_i , \forall i \in [k],$ and time horizon $T$, such that for all algorithms for the online multi-scale experts problem with symmetric range,   there is a sequence of     $T$ gain vectors such that 
	\[\textstyle
	\exists i \in \actions:~~\ex{\regret{i}} >  \frac{ \range_{i}}{4}\cdot\sqrt{T\log (k\cdot \frac{\range_{\max}}{\range_{\min}})} 
	\]
\end{theorem} 


We then show the following upper bound for the bandit setting when the range is symmetric. This bound also follows the same style of action-specific regret bounds as in Theorem~\ref{thm:bandit-regret-bound}.
\begin{theorem}
	\label{thm:bandit-regret-bound-symmetric}
There exists an algorithm for the multi-scale bandits problem with symmetric range that	takes as input 
the ranges $c_i, ~\forall ~i \in A$, and a parameter $0 < \epsilon \le 1/2$,  and satisfies:
	\begin{equation}
	\label{eqn:bandit-regret-bound-symmetric}
	\textstyle
	\forall i \in \actions:~~\ex{\regret{i}}\leq O \big( \epsilon T + \frac{k}{\epsilon} \frac{c_{\max}}{c_{\min}} \log\big(\frac{k}{\epsilon} \frac{c_{\max}}{c_{\min}}\big) \big) \cdot c_i .
	\end{equation}
\end{theorem}


Also, similar to Section~\ref{sec:multi-scale-nonnegative-results}, we can compute the pure-additive version of the bound in Theorem~\ref{thm:bandit-regret-bound-symmetric} by setting $\epsilon=\sqrt{\frac{k\frac{\range_{\max}}{\range_{\min}}\log (kT\cdot \frac{\range_{\max}}{\range_{\min}})}{T}}$, as in Corollary~\ref{cor:expert-regret-bound-rewards-additive}. This bound is comparable to the standard regret bound of $O(\range_{\max}\sqrt{kT\log k})$~\citep{auer1995gambling} for the adversarial multi-armed bandits problem.
\begin{corollary}
	\label{cor:bandits-regret-bound-rewards-additive-symmetric}
	There exists an algorithm for the online multi-scale bandits problem with symmetric range that satisfies:
	\begin{equation}\textstyle
	\forall i \in \actions:~~\ex{\regret{i}} \leq {O} \left( \range_{i}\cdot\sqrt{Tk\cdot \frac{\range_{\max}}{\range_{\min}}\log (k T\cdot \frac{\range_{\max}}{\range_{\min}})} \right) .
	\end{equation}
\end{corollary}

Once again, for the bandit problem, the following theorem shows  that this bound cannot be improved beyond logarithmic factors
(to get a  guarantee like that of Theorem \ref{thm:bandit-regret-bound}, for instance).

\begin{theorem}\label{thm:bandit-lower-bound-symmetric} 
	There exists an action set of size $k$, and ranges $c_i , \forall i \in [k],$ such that for all algorithms for the online multi-scale bandit problem with symmetric range, for all sufficiently large time horizon $T$,  there is a sequence of $T$ gain vectors such that 
	\[
		\exists i \in \actions:~~\ex{\regret{i}} >   \frac{\range_{i}}{8\sqrt{2}}\cdot\sqrt{Tk\cdot \frac{\range_{\max}}{\range_{\min}}} . 
	\]
\end{theorem}



\subsection{Upper bound for experts with symmetric range - Proof of Theorem~\ref{thm:expert-regret-bound-symmetric}}
\label{sec:proof-thm-1}
Recall the proof of Proposition~\ref{prop:expert-regret-bound-general}. The proof only requires $\rew{i}{t}\in [-\range_i,\range_i]$ for all $i\in \actions, t\in[T]$.  Choosing $q$ to be $\mathbf{1}_i$, a vector with a $1$-entry in $i^{\textrm{th}}$ coordinate and $0$-entries elsewhere for an action $i \in \actions$, and noting that 
\[
\textstyle
\sum_{t \in [T]} \sum_{i\in \actions}\p{i}{t} \frac{(\rew{i}{t})^2}{\range_i} \le \sum_{t \in [T]} \sum_{i\in \actions}\p{i}{t} \cdot \big| \rew{i}{t} \big| ~,
\]
we get the following regret bound as a corollary of Proposition~\ref{prop:expert-regret-bound-general}.

\begin{corollary}
\label{lem:expert-regret-bound}
For any initial distribution $\mu$ over $\actions$, and any learning rate parameter $0 < \eta \le 1$, the MSMW algorithm achieves the following regret bound:
\begin{equation}
\textstyle
\forall i\in\actions:~~\ex{\regret{i}} \leq \eta\cdot\ex{\sum_{t \in [T]} \big| \rew{i}{t} \big|}+\frac{1}{\eta} \range_i \cdot \log \big( \frac{1}{\mu_i} \big)+ \frac{1}{\eta} \sum_{j\in\actions} \mu_j \range_j
\end{equation}
\end{corollary}
Now, we can prove the multi-scale regret upper-bound in Theorem~\ref{thm:expert-regret-bound-symmetric} using Corollary~\ref{lem:expert-regret-bound}.

\begin{proof}[ \textbf{of Theorem \ref{thm:expert-regret-bound-symmetric}}]
The proof follows by choosing an appropriate initial distribution $\mu$ in  Corollary~\ref{lem:expert-regret-bound}.
By Corollary~\ref{lem:expert-regret-bound}, we have:
\[
\textstyle
\ex{\regret{i}}\leq \eta\cdot\ex{\sum_{t \in [T]} \big| \rew{i}{t} \big|}+\frac{1}{\eta} \range_i \cdot \log(\frac{1}{\mu_i})+ \frac{1}{\eta} \sum_{j\in\actions} \mu_j \range_j
\]

Let $\imin$ be an action with the minimum range $\range_{\imin} = \range_{\min}$.
Consider an initial distribution $\mu_j = \pi_j \frac{\range_{\min}}{\range_j}$ for all $j \ne \imin$, and $\mu_{\imin} = 1 - \sum_{j \ne \imin} \mu_j$, i.e., putting all remaining probability mass on action $\imin$.
Then, the third term on the RHS is upper bounded by:
\[
\textstyle
\sum_{j \in \actions} \mu_j \range_j = \sum_{j \ne \imin} \mu_j \range_j + \mu_{\imin} \range_{\imin} = \sum_{j \ne \imin} \pi_j \range_{\min} + \mu_{\imin} \range_{\min} \le 2 \range_{\min} \le 2 \range_i \enspace.
\]
For $i \ne \imin$, by the definition of $\mu_i$, we have:
\begin{align*}
\ex{\regret{i}} & 
\textstyle
\leq \eta\cdot\ex{\sum_{t \in [T]} \big| \rew{i}{t} \big|} +\frac{1}{\eta} \range_i \cdot \log(\frac{1}{\pi_i} \cdot \frac{\range_i}{\range_{\min}})+ \frac{1}{\eta} \cdot 2 \range_{\min} \\
& 
\textstyle
= \eta\cdot\ex{\sum_{t \in [T]} \big| \rew{i}{t} \big|} + O \bigg( \frac{1}{\eta} \log \big( \frac{1}{\pi_i} \cdot \frac{\range_i}{\range_{\min}} \big) \cdot \range_i \bigg) \enspace.
\end{align*}
So the theorem follows by choosing $\eta = \epsilon$. For $i = \imin$, note that $\mu_j \le \pi_j$ for all $j \ne \imin$ and, thus, $\mu_{\imin} = 1 - \sum_{j \ne \imin} \mu_j \ge 1 - \sum_{j \ne \imin} \pi_j = \pi_{\imin} = \pi_{\imin} \frac{\range_{\min}}{\range_{\imin}}$.
The theorem then holds following the same calculation as in the $j \ne \imin$ case.
\end{proof}

\subsection{Lower bound for experts with symmetric range - proof of Theorem~\ref{thm:expert-log-dependency-symmetric}}
\label{sec:proof-log-dep}
\begin{proof}[\textbf{of Theorem~\ref{thm:expert-log-dependency-symmetric}}]
We first show that for any online learning algorithm, and any sufficiently large $h > 1$, there is an instance that has two experts with $\range_1 = 1$ and $\range_2 = h$ with $T = \Theta(\log{h})$ rounds, such that either
	\[
	\ex{\regret{1}} > \tfrac{1}{2} T + \sqrt{h}~,\quad\quad\textrm{or}\quad\quad\ex{\regret{2}} > \tfrac{1}{2} T h + \tfrac{1}{5} h \log_2{h}~.
	\]
We will construct this instance with $T = \frac{1}{2} \log_2{h} - 1$ rounds adaptively that always has gain $0$ for action $1$ and gain either $h$ or $-h$ for action $2$. The proof of the theorem then follows as $\range_{\min}=1$, $\range_{\max}=h$, $T= \frac{1}{2} \log_2{h} - 1$, and $k=2$ in this instance.
Let $q_t$ denote the probability that the algorithm picks action $2$ in round $t$ after having the same rewards $1$ and $h$ for the two actions respectively in the first $t-1$ rounds. 
We will first show that (1) if the algorithm has small regret with respect to action $1$, then $q_t$ must be upper bounded since the adversary may let action $2$ have cost $-h$ in any round $t$ in which $q_t$ is too large.
Then, we will show that (2) since $q_t$ is upper bounded for any $1 \le t \le T$, the algorithm must have large regret with respect to action $2$. 

We proceed with the upper bounding $q_t$'s.
Concretely, we will show the following lemma.

\begin{lemma}
\label{lem:log-dependency}
Suppose $\ex{\regret{1}} \le \tfrac{1}{2} T + \sqrt{h}$.
Then, for any $1 \le t \le T$, we have $q_t \le \frac{2^{t}}{\sqrt{h}}$.
\end{lemma}

\begin{proof}[Proof of Lemma~\ref{lem:log-dependency}]
We will prove by induction on $t$.
Consider the base case $t = 1$. 
Suppose for contradiction that $q_1 > \frac{2}{\sqrt{h}}$.
Then, consider an instance in which action $2$ always has gain.
In this case, the expected gain of the algorithm (even if it always correctly picks action $1$ in the remaining instance) is at most $q_1 \cdot (-h) < - 2 \sqrt{h}$.
This is a contradiction to the assumption that $\ex{\regret{1}} \le \frac{1}{2} T + \sqrt{h} < 2 \sqrt{h}$.

Next, suppose the lemma holds for all rounds prior to round $t$.
Then, the expected gain of algorithm in the first $t-1$ rounds if arm $2$ has gain $H$ is 
\[
\sum_{\ell = 1}^{t-1} q_{\ell} \cdot h \le \sum_{\ell = 1}^{t-1} 2^{\ell} \sqrt{h} = \big( 2^{t} - 2 \big) \sqrt{h} ~.
\]
Suppose for contradiction that $q_t > \frac{2^{t}}{\sqrt{h}}$.
Then, consider an instance in which action $2$ has gain $H$ in the first $t-1$ rounds and $-H$ afterwards.
In this case, the expected gain of the algorithm (even if it always correctly picks action $1$ after round $t$) is at most 
\[
\big( 2^{t} - 2 \big) \sqrt{h} + q_t (-h) < \big( 2^{t} - 2 \big) \sqrt{h} + 2^{t} \sqrt{h} < - 2 \sqrt{h} ~.
\]
This is a contradiction to the assumption that $\ex{\regret{1}} \le \frac{1}{2} T + \sqrt{h} < 2 \sqrt{h}$.
\end{proof}

Consider an instance in which action $2$ always has gain $H$.
Suppose that $\ex{\regret{1}} \le \frac{1}{2} T + \sqrt{h}$.
As an immediate implication of the above lemma, the algorithm  is that the expected gain of the algorithm is upper bounded by:
\[
\sum_{t=1}^T q_t h \le \sum_{t=1}^T 2^{t} \sqrt{h} < 2^{T+1} \sqrt{h} = h ~.
\]

Note that in this instance $\ex{\reward{2}} = T \cdot h$.
Thus, the regret w.r.t.\ action $2$ is at least $(T - 1) h$, which is greater than $\tfrac{1}{2} \cdot \ex{\reward{2}} + \tfrac{1}{5} h \log_2{h}$ for sufficiently large $h$.
\end{proof}

\subsection{Upper bound for bandits with symmetric range - Proof of Theorem~\ref{thm:bandit-regret-bound-symmetric}}


We start by presenting the following regret bound, whose proof is an alteration of that for Lemma~\ref{lem:bandit-msmw} under symmetric range. Next, we prove Theorem~\ref{thm:bandit-regret-bound-symmetric}.
\begin{lemma}
\label{lem:bandit-msmw-symmetric}
For any exploration rate $0 < \gamma \leq \min \{\frac{1}{2}, \frac{c_{\min}}{c_{\max}} \}$ and any learning rate $0 < \eta \le \frac{\gamma}{k}$, the Bandit-MSMW algorithm (Algorithm~\ref{alg:Bandit-MSMW}) achieves the following regret bound:
%
\[
\forall i\in\actions: \ex{\regret{i}}\leq O \left( \frac{1}{\eta} \log \big( \frac{k}{\gamma} \big) \cdot \range_i + \gamma T \cdot c_{\max} \right)
\]
\end{lemma}

\begin{proof}[\textbf{of Lemma~\ref{lem:bandit-msmw-symmetric}}]
We further define:
\begin{eqnarray*}
\talg & \triangleq & \textstyle \sum_{t\in[T]} \rew{i_t}{t} = \sum_{t\in[T]}\tpb{t}\cdot\trewb{t}~,\\
\treward{j} & \triangleq & \textstyle \sum_{t\in[T]}\trew{j}{t}~.
\end{eqnarray*}
%

In expectation over the randomness of the algorithm, we have:
\begin{enumerate}
	\item $\ex{\alg} = \ex{\talg}$; and
	\item $\reward{j} = \ex{\treward{j}}$ for any $j \in \actions$.
\end{enumerate}
%

Hence, to upper bound $\ex{\regret{i}}=\reward{i}-\ex{\alg}$, it suffices to upper bound $\ex{\treward{i} - \talg}$.

By the definition of the probability that the algorithm picks each arm, i.e., $\tpb{t}$, and that reward of each round is at least $-c_{\max}$, we have that:
\[
\ex{\talg} \ge (1 - \gamma) \sum_{t\in[T]} \pb{t} \cdot \trewb{t} - \gamma T c_{\max}~.
\]

Hence, for any benchmark distribution $\qb$ over $\actions$, we have that:
\begin{align}
\textstyle
\sum_{j \in \actions} q_j \cdot \ex{\treward{j}} - \ex{\talg} \le ~ & \textstyle \ex{\sum_{j \in \actions} q_j \cdot \treward{j} - \sum_{t\in[T]} \pb{t} \cdot \trewb{t}} + \frac{\gamma}{1-\gamma} \ex{\talg} + \frac{\gamma}{1-\gamma} T c_{\max} \notag \\
\le ~ & \textstyle \ex{\sum_{j \in \actions} q_j \cdot \treward{j} - \sum_{t\in[T]} \pb{t} \cdot \trewb{t}} + 2 \gamma \ex{\talg} + 2 \gamma T c_{\max} \notag \\
\le ~ & \textstyle \ex{\sum_{j \in \actions} q_j \cdot \treward{j} - \sum_{t\in[T]} \pb{t} \cdot \trewb{t}} + 4 \gamma T c_{\max} ~. \label{eqn:bandit-msmw-symmetric1}
\end{align}
where the 2nd inequality is due to $\gamma \le \frac{1}{2}$, and the 3rd inequality follows by that $c_{\max}$ is the largest possible reward per round.

Next, we upper bound the 1st term on the RHS of \eqref{eqn:bandit-msmw-symmetric1}.
Note that $\pb{t}$'s are the probability of choosing experts by MSMW when the experts have rewards $\trewb{t}$'s.
By Proposition~\ref{prop:expert-regret-bound-general}, we have that for any benchmark distribution $\qb$ over $S$, the Bandit-MSMW algorithm satisfies that:
\begin{equation}
\label{eqn:bandit-msmw-symmetric2}
\sum_{j \in \actions} q_j \cdot \treward{j} - \sum_{t\in[T]} \pb{t} \cdot \trewb{t} \leq \eta \sum_{t\in[T]} \sum_{j \in \actions} \frac{\p{j}{t}}{\range_j} \cdot \big( \trew{j}{t} \big)^2 + \frac{1}{\eta} \sum_{j \in \actions} \range_j \left( q_j \ln \big( \frac{q_j}{\p{j}{1}} \big) - q_j + \p{j}{1} \right) ~.
\end{equation}

For any $t \in [T]$ and any $j \in \actions$, by the definition of $\trew{j}{t}$, it equals $\tfrac{\rew{j}{t}}{\tp{j}{t}}$ with probability $\tp{j}{t}$, and equals $0$ otherwise.
Thus, if we fix the random coin flips in the first $t-1$ rounds and, thus, fix $\tpb{t}$, and take expectation over the randomness in round $t$, we have that:
\[
\ex{\frac{\p{j}{t}}{\range_j} \cdot \big( \trew{j}{t} \big)^2} = \frac{\p{j}{t}}{\range_j} \cdot \tp{j}{t} \cdot \left(\frac{\rew{j}{t}}{\tp{j}{t}} \right)^2 = \frac{\p{j}{t}}{\tp{j}{t}} \frac{(\rew{j}{t})^2}{\range_j} ~.
\]

Further note that $\tp{j}{t} \ge (1 - \gamma) \p{j}{t}$, and $| \rew{j}{t} | \le \range_j$, the above is upper bounded by $\frac{1}{1 - \gamma} | \rew{j}{t} | \le 2 | \rew{j}{t} | \le 2 c_{\max}$.
Putting together with \eqref{eqn:bandit-msmw-symmetric2}, we have that for any $0 < \eta \le \frac{\gamma}{n}$:
\begin{align*}
\ex{\sum_{j \in \actions} q_j \cdot \treward{j} - \sum_{t\in[T]} \pb{t} \cdot \trewb{t}}
\le ~
&
\eta \sum_{t\in[T]}\sum_{j \in \actions} 2 c_{\max} + \frac{1}{\eta} \sum_{j \in \actions} \range_j \left( q_j \ln \big( \frac{q_j}{\p{j}{1}} \big) - q_j + \p{j}{1} \right) \\
= ~
&
2 \eta T k c_{\max} + \frac{1}{\eta} \sum_{j \in \actions} \range_j \left( q_j \ln \big( \frac{q_j}{\p{j}{1}} \big) - q_j + \p{j}{1} \right)
\end{align*}

Combining with \eqref{eqn:bandit-msmw-symmetric1}, we have (recall that $\eta \le \frac{\gamma}{k}$):
\begin{align*}
\sum_{j \in \actions} q_j \cdot \ex{\treward{j}} - \ex{\talg} & \le 2 \eta T k c_{\max} + \frac{1}{\eta} \sum_{j \in \actions} \range_j \left( q_j \ln \big( \frac{q_j}{\p{j}{1}} \big) - q_j + \p{j}{1} \right) + 4 \gamma T c_{\max} \\
& \le \frac{1}{\eta} \sum_{j \in \actions} \range_j \left( q_j \ln \big( \frac{q_j}{\p{j}{1}} \big) - q_j + \p{j}{1} \right) + 6 \gamma T c_{\max}
\end{align*}

Let $\qb = (1 - \gamma) \mathbf{1}_{i} + \frac{\gamma}{k} \mathbf{1}$.
Recall that $\pb{1} = (1 - \gamma) \mathbf{1}_{\imin} + \frac{\gamma}{k} \mathbf{1}$ (recall $\imin$ is the arm with minimum range $\range_{\imin}$).
Similar to the discussion for the expert problem in Section~\ref{sec:expert-reward-only}, the 1st term on the RHS is upper bounded by $O \left( \frac{1}{\eta} \log \big( \frac{k}{\gamma} \big) \cdot \range_i \right)$.
Hence, we have:
\begin{equation}
\label{eqn:bandit-msmw-symmetric3}
\sum_{j \in \actions} q_j \cdot \ex{\treward{j}} - \ex{\talg} \le O \left( \frac{1}{\eta} \log \big( \frac{k}{\gamma} \big) \cdot \range_i \right) + 6 \gamma T c_{\max} ~.
\end{equation}

Further, the LHS is lower bounded as:
\[ 
(1 - \gamma) \ex{\treward{i}} + \frac{\gamma}{k} \sum_{j \in \actions} \ex{\treward{j}} - \ex{\talg} \ge (1 - \gamma) \ex{\treward{i}} - \gamma T c_{\max} - \ex{\talg} ~.
\]

The lemma then follows by putting it back to \eqref{eqn:bandit-msmw-symmetric3} and rearranging terms.
\end{proof}

\begin{proof}[ \textbf{of Theorem~\ref{thm:bandit-regret-bound-symmetric}}]
Let $\gamma = \epsilon \frac{c_{\min}}{c_{\max}}$ and $\eta = \frac{\gamma}{k}$ in Lemma~\ref{lem:bandit-msmw-symmetric}. Theorem follows noting that $\gamma c_{\max} = \epsilon c_{\min} \le \epsilon c_i$.
\end{proof}

\subsection{Lower-bound for bandits with symmetric range - Proof of Theorem~\ref{thm:bandit-lower-bound-symmetric}}
\input{lowerbound}

%% file: lowerbound.tex
\label{sec:bandit-lower-bound-symmetric}
\begin{proof}[\textbf{of Theorem~\ref{thm:bandit-lower-bound-symmetric}}]
We first show that  for any online multi-scale bandits algorithm problem, and there is an instance that has two arms with $\range_1 = 1$ and $\range_2 = h$ for some sufficiently large $h$, a sufficiently large $T$,  
and $\epsilon  = \sqrt{\frac{h}{256 T}}$, such that either
	\[
	\ex{\regret{1}} > \epsilon T + \tfrac{1}{256 \epsilon} h ~,\quad\quad \textrm{or}\quad\quad \ex{\regret{2}} > \epsilon T h + \tfrac{1}{256 \epsilon} h^2
	\]
We will prove the existence of this instance by looking at the stochastic setting, i.e., the gain vectors $\rewb{t}$'s are i.i.d.\ for $1 \le t \le T$.	
We consider two instances, both of which admit a fixed gain of $0$ for action $1$.
In the first instance, the gain of action $2$ is $h$ with probability $\frac{1}{2} - 2 \epsilon$, and $-h$ otherwise.
Hence, the expected gain of playing action $2$ is $- 4 \epsilon h$ per round in instance $1$.
In the second instance, the gain of action $2$ is $h$ with probability $\frac{1}{2} + 2 \epsilon$, and $-h$ otherwise.
Hence, the expected gain of playing action two is $4 \epsilon h$ per round in instance $2$. Note this proves the theorem, as $\range_{\min}=1$, $\range_{\max}=h$, $k=2$ and and $T = \frac{h}{256 \epsilon^2}$.

Suppose for contradiction that the algorithm satisfies:
\[
\ex{\regret{1}} \le \epsilon T + \tfrac{1}{256 \epsilon} h = \tfrac{1}{128 \epsilon} h \quad,\quad \ex{\regret{2}} \le \epsilon h T + \tfrac{1}{256 \epsilon} h^2 = \tfrac{1}{128 \epsilon} h^2 ~.
\]

Let $N_1$ denote the expected number of times that the algorithm plays action $2$ in instance $1$.
Then, the expected regret with respect to action $1$ in instance $1$ is $N_1 \cdot 4 \epsilon h$.
By the assumption that $\ex{\regret{1}} \le \tfrac{1}{128 \epsilon} h$, we have $N_1 \le \frac{1}{512 \epsilon^2}$.

Next, by standard calculation, we get that the Kullback-Leibler (KL) divergence of the observed rewards in a single round in the two instances is $0$ if action $1$ is played and is at most $64 \epsilon^2$ (for $0 < \epsilon < 0.1$) if action $2$ is played.
So the KL divergence of the observed reward sequences in the two instances is at most $64 \epsilon^2 \cdot N_1 \le \frac{1}{8}$.

Then, we use a standard inequality about KL divergences.
For any measurable function $\psi : X \mapsto \{1, 2\}$, we have $\Pr_{X \sim \rho_1} \big(\psi(X) = 2\big) + \Pr_{X \sim \rho_2} \big(\psi(X) = 1\big) \ge \frac{1}{2} \exp{\big( - KL(\rho_1, \rho_2) \big)}$. For any $1 \le t \le T$, let $\rho_1$ and $\rho_2$ be the distribution of observed rewards up to a round $t$ in the two instances, and let $\psi(X)$ be the action played by the algorithm.
By this inequality and the above bound on the KL divergence between the observed rewards in the two instances, we get that in each round, the probability that the algorithm plays action $2$ in instance $1$, plus the probability that the algorithm plays action $1$ in instance $2$, is at least $\frac{1}{2} \exp{(- \frac{1}{8})} > \frac{2}{5}$ in any round $t$.
Thus, the expected number of times that the algorithm plays action $1$ in instance $2$ from round $1$ to $T$, denoted as $N_2$, is at least $N_2 \ge \frac{2}{5} \cdot T - N_1 \ge \frac{1}{3} \cdot T$, where the second inequality holds for sufficiently large $h$.
Therefore, the expected regret w.r.t.\ action $2$ in instance $2$ is at least:
$ 4\epsilon h \cdot \frac{1}{3} \cdot T = \frac{4}{3} \epsilon h T > \tfrac{1}{128 \epsilon} h^2$. This is a contradiction to our assumption that $\ex{\regret{2}} \le \tfrac{1}{128 \epsilon} h^2$.
\end{proof}

%% file: conclusion.tex
\section{Conclusion}

Revenue management has emerged as a competitive toolbox of strategies for increasing the profit of web-based markets. In particular, dynamic pricing, and dynamic auction design as its less mature relative, have become prevalent market mechanisms in nearly all industries. In this paper, we studied these problems from the perspective of online learning. For the online auction for single buyer, we showed regret bounds that scale with the best fixed price, rather than the range of the values (with a generalization to learning auctions). Moreover, we demonstrated a connection between the optimal regret bounds for this problem and offline sample complexity lower-bounds of approximating optimal revenue, studied in \cite{ColeR14,HuangMR15}. Using this connection, we showed our regret bounds are almost optimal as they match these information theoretic lower-bounds. We further generalized our result to online pricing (bandit feedback) and online auction with multiple-buyers. 

The key to our development and improved regret bounds for online auction design is generalizing the classical learning from experts and multi-armed bandit problems to their ``multi-scale versions'', where the reward of each action is in a different range. Here the objective is to design online learning algorithms whose regret with respect to a given action scales with its own range, rather than the maximum range. We showed how a variant of online mirror descent solves this learning problem.

%% file: appendix.tex
\section{Other Deferred Proofs and Discussions}
\subsection{Discussion on choice of \texorpdfstring{$\pib$ for bandit symmetric range}{Lg}}
\label{appendix:bandit-symmetric-initial}
We now describe how the choice of initial distribution $\pib$ affects the bound given in Theorem~\ref{thm:expert-regret-bound-symmetric}. 
\begin{itemize}
	\item When the action set is finite, we can choose $\pib$ to be the uniform distribution to get the term 
	$$ O \big( \frac{1}{\epsilon} \log ( \myfrac{k \range_i}{\range_{\min}} ) \cdot \range_i \big)$$ This recovers the standard bound by setting $\range_i = \range_{\max}$ for all $i\in \actions$. 
	\item  We can choose $\pi_i = \frac{\range_i}{\sum_{j \in \actions}  \range_j}$ to get 
	$O \big( \frac{1}{\epsilon} \log ( \myfrac{\sum_{j\in \actions} \range_j}{\range_{\min}} ) \cdot \range_i \big)$. 
	In particular, if the $\range_i$'s form an arithmetic progression with a constant difference then this is just 
	$O \big( \frac{\log{k}}{\epsilon} \cdot \range_{i}  \big)$.
\end{itemize}

\subsection{Proof of Proposition~\ref{prop:expert-regret-bound-general} from first principles}
\label{appendix:OMD-first-principles}

We also provide an elementary proof of this lemma using first principles. 

\begin{proof}[\textbf{of Proposition~\ref{prop:expert-regret-bound-general}}]
Based on the update rule of Algorithm~\ref{alg:MSMW}, we have $\rew{i}{t}=\frac{\range_i}{\eta}\log(\frac{\w{i}{t+1}}{\p{i}{t}})$ for any $i \in \actions$. 
Therefore:
\begin{align}
\rewb{t}\cdot \big( \qb - \pb{t} \big) & = \sum_{i\in\actions}\rew{i}{t} \big( q_i-\p{i}{t} \big) \notag \\
& = \sum_{i\in\actions} \frac{\range_i}{\eta} \cdot \log \big( \frac{\w{i}{t+1}}{\p{i}{t}} \big) \cdot \big( q_i -\p{i}{t} \big) \notag \\
& = \frac{1}{\eta} \left( \sum_{i \in S} \range_i \cdot q_i \cdot \log \big( \frac{\w{k}{t+1}}{\p{k}{t}} \big) + \sum_{i\in\actions}\range_i\cdot\p{i}{t} \cdot \log \big(\frac{\p{i}{t}}{\w{i}{t+1}}\big)\right) \notag \\
& = \frac{1}{\eta} \bigg( \sum_{i \in S} \range_i \cdot q_i \cdot \log \big( \frac{\w{k}{t+1}}{\p{k}{t+1}} \big) + \sum_{i \in S} \range_i \cdot q_i \cdot \log \big( \frac{\p{k}{t+1}}{\p{k}{t}} \big) \notag \\
& \quad\quad\quad + \sum_{i\in\actions}\range_i\cdot\p{i}{t} \cdot \log \big(\frac{\p{i}{t}}{\w{i}{t+1}}\big) \bigg) 
\label{eq:tel1}
\end{align}

Now, note that due to the normalization step of Algorithm~\ref{alg:MSMW}, for any $i \in S$ we have:
%
\[
\range_i \cdot \log(\frac{\w{i}{t+1}}{\p{i}{t+1}}) = \lambda = \sum_{j\in\actions} \range_j \cdot \p{j}{t+1}\cdot\frac{\lambda}{\range_j} = 
\sum_{j\in\actions} \range_j\cdot\p{j}{t+1} \cdot \log(\frac{\w{j}{t+1}}{\p{j}{t+1}}) 
\]

So the first summation in \eqref{eq:tel1} is equal to:
\begin{align}
\sum_{i \in S} \range_i \cdot q_i \cdot \log \big( \frac{\w{k}{t+1}}{\p{k}{t+1}} \big)
& =
\sum_{i \in S} q_i \cdot \sum_{j\in\actions} \range_j\cdot\p{j}{t+1} \cdot \log(\frac{\w{j}{t+1}}{\p{j}{t+1}}) \notag \\
& =
\sum_{j\in\actions} \range_j\cdot\p{j}{t+1} \cdot \log(\frac{\w{j}{t+1}}{\p{j}{t+1}}) \notag \\
& = 
\sum_{i\in\actions} \range_i\cdot\p{i}{t+1} \cdot \log(\frac{\w{i}{t+1}}{\p{i}{t+1}}) 
\label{eq:tel2}
\end{align}

Combining Eqn.~\eqref{eq:tel1} and \eqref{eq:tel2}, we have: 
\begin{align*}
\rewb{t}\cdot \big( \qb - \pb{t} \big) & = \frac{1}{\eta} \sum_{i\in\actions}\range_i\cdot\left(\p{i}{t}\cdot\log(\frac{\p{i}{t}}{\w{i}{t+1}}) + \p{i}{t+1}\cdot\log(\frac{\w{i}{t+1}}{\p{i}{t+1}})\right) \\
& \quad\quad\quad + \frac{1}{\eta} \sum_{i \in S} \range_i \cdot q_i \cdot \log \big( \frac{\p{i}{t+1}}{\p{i}{t}} \big)
\end{align*}

The 2nd part is a telescopic sum when we sum over $t$. 
We will upper bound the 1st part as follows.
By $\log(x)\leq (x-1)$, we get that:
\begin{align*}
& \sum_{i\in\actions}\range_i\cdot\left(\p{i}{t}\cdot\log(\frac{\p{i}{t}}{\w{i}{t+1}}) + \p{i}{t+1}\cdot\log(\frac{\w{i}{t+1}}{\p{i}{t+1}})\right) 
\\
& \quad\quad\quad\quad
\le
\sum_{i\in\actions}\range_i\cdot\left(\p{i}{t}\cdot\log(\frac{\p{i}{t}}{\w{i}{t+1}})-\p{i}{t+1}+\w{i}{t+1}\right) \\
& \quad\quad\quad\quad
= 
\sum_{i\in\actions}\range_i\cdot \big( \p{i}{t}-\p{i}{t+1} \big) + \sum_{i\in\actions}\range_i\cdot\left(\p{i}{t}\cdot\log(\frac{\p{i}{t}}{\w{i}{t+1}})-\p{i}{t}+\w{i}{t+1}\right)
\end{align*}

Again, the 1st part is a telescopic sum when we sum over $t$. 
We will further work on the 2nd part.
By the relation between $\w{i}{t+1}$ and $\p{i}{t}$, we get that:
\[
\sum_{i\in\actions}\range_i\cdot\left(\p{i}{t}\cdot\log(\frac{\p{i}{t}}{\w{i}{t+1}})-\p{i}{t}+\w{i}{t+1}\right) = \sum_{i\in\actions}\range_i\cdot\p{i}{t}\left(-\eta\cdot\frac{\rew{i}{t}}{\range_i}-1+\exp(\eta\cdot\frac{\rew{i}{t}}{\range_i})\right)
\]

Note that $\eta \cdot \frac{\rew{i}{t}}{\range_i} \in [-1, 1]$ because $\rew{i}{t} \in [-\range_i, \range_i]$ and $0 < \eta \le 1$.
By $\exp(x)-x-1\leq x^2$ for $-1 \le x \le 1$ and that $\eta \rew{i}{t} \in [-\range_i, \range_i]$, the above is upper bounded by $\eta^2\sum_{i\in \actions}\p{i}{t} \frac{(\rew{i}{t})^2}{\range_i}$.
Putting together, we get that:
\[
\rewb{t}\cdot \big( \qb - \pb{t} \big)
\le 
\frac{1}{\eta} \sum_{i \in S} \range_i \cdot \bigg( q_i \cdot \log \big( \frac{\p{i}{t+1}}{\p{i}{t}} \big) + \p{i}{t} - \p{i}{t+1} \bigg) 
+ \eta \sum_{i\in \actions}\p{i}{t} \frac{(\rew{i}{t})^2}{\range_i}
\]

Summing over $t$, we have:
\[
\rewb{t}\cdot \big( \qb - \pb{t} \big) \leq \frac{1}{\eta} \sum_{i \in S} \range_i \cdot \bigg( q_i \cdot \log \big( \frac{\p{i}{T+1}}{\p{i}{1}} \big) + \p{i}{1} - \p{i}{T+1} \bigg) + \eta \sum_{t\in[T]} \sum_{i\in \actions}\p{i}{t} \frac{(\rew{i}{t})^2}{\range_i}
\]

Finally, by $\log(x) \leq (x-1)$, we get that $q_i \log \big( \frac{\p{i}{T+1}}{q_i} \big) \le \p{i}{T+1} - q_i$.
Hence, we have:
\[
\rewb{t}\cdot \big( \qb - \pb{t} \big) \leq \frac{1}{\eta} \sum_{i \in S} \range_i \cdot \bigg( q_i \cdot \log \big( \frac{q_i}{\p{i}{1}} \big) + \p{i}{1} - q_i \bigg) + \eta \sum_{t\in[T]} \sum_{i\in \actions}\p{i}{t} \frac{(\rew{i}{t})^2}{\range_i}
\]

The lemma then follows by our choice of the initial distribution.
\end{proof}

\subsection{Proof of OMD regret bound}
\label{appendix:OMD-proof}
In order to prove the OMD regret bound, we need some properties of Bregman divergence. 
\begin{lemma}[Properties of Bregman divergence~\citep{bubeck2011introduction}]
\label{lem:BD-properties}
Suppose $F(\cdot)$ is a Legendre function and $D_F(\cdot,\cdot)$ is its associated Bregman divergence as defined in Definition~\ref{def:BD}. Then: 
\begin{itemize}
 \item $D_F(x,y)>0$ if $x\neq y$ as $F$ is strictly convex, and $D_F(x,x)=0$. 
 \item $D_F(.,y)$ is a convex function for any choice of $y$.
 \item  (\emph{Pythagorean theorem}) If $\mathcal{A}$ is a convex set, $a\in\mathcal{A}$, $b\notin\mathcal{A}$ and $c=\argmin{x\in\mathcal{A}}{D_F(x,b)}$, then 
\[
D_F(a,c)+D_F(c,b)\leq D_F(a,b)
\]
\end{itemize}
\end{lemma}
Given Lemma~\ref{lem:BD-properties}, we are now ready to prove Lemma~\ref{lem:OMD-regret}.
\begin{proof}[Proof of Lemma~\ref{lem:OMD-regret}]
To obtain the OMD regret bound, we have:
\begin{align}
\qb\cdot\rewb{t}-\pb{t}\cdot\rewb{t}
&=\frac{1}{\eta}(\qb-\pb{t})\cdot(\nabla F(\wb{t+1})-\nabla F(\pb{t}))\nonumber\\
&=\frac{1}{\eta}(D_F(qb,\pb{t})+D_F(\pb{t},\wb{t+1})-D_F(qb,\wb{t+1}))\nonumber\\
&\label{eq:neg-ent-OMD}
\overset{(1)}{\leq} \frac{1}{\eta} D_F(\pb{t},\wb{t+1})+\frac{1}{\eta}\left(D_F(\qb,\pb{t})-D_F(\qb,\pb{t+1})\right)
\end{align}
where in (1) we use $D_F(\pb{t+1},\wb{t+1})\geq 0$ and $D_F(\qb,\pb{t+1})+D_F(\pb{t+1},\wb{t+1})\leq D_F(\qb,\wb{t+1})$ due to Pythagorean theorem (Lemma~\ref{lem:BD-properties}). By summing up both hand sides of (\ref{eq:neg-ent-OMD}) for $t=1,\cdots,T$ we have:
\begin{equation}
\sum_{t\in[T]} \rewb{t} \cdot \big( \qb - \pb{t} \big)\leq \frac{1}{\eta}\sum_{t\in[T]}D_F(\pb{t},\wb{t+1})+\frac{1}{\eta}D_F(\qb,\pb{1})
\end{equation}
\end{proof}